\documentclass[a4paper,11pt]{article}

\usepackage{microtype}
\usepackage{mathtools,amsmath}\usepackage{amssymb,amsthm}
\usepackage[T1]{fontenc}
\usepackage{graphicx} \usepackage{xspace} \usepackage{xypic}
\usepackage[numbers]{natbib} \usepackage{hyperref,color}\usepackage{pdfpages}
\usepackage{fullpage}
\usepackage[mathlines]{lineno}
\usepackage{algorithm}
\usepackage[noend]{algpseudocode}

\newtheorem{theorem}{Theorem} 
\newtheorem{lemma}[theorem]{Lemma}

\newcommand{\rectset}{\mathsf{R}}

\newcommand{\ml}{\ensuremath{\lambda}}
\newcommand{\dom}{\mathsf{D}}
\newcommand{\rdom}{\mathsf{C}}
\newcommand{\rangebbst}{\mathcal{T}_\mathsf{ran}}
\newcommand{\xydom}{\mathsf{D}_{xy}}
\newcommand{\xdom}{\mathsf{D}_{x}}
\newcommand{\ydom}{\mathsf{D}_{y}}

\newcommand{\xybd}{\partial_{xy}}
\newcommand{\realp}{\mathbb{R}^2}

\def\origi{\textsf{originate}\xspace}
\def\termi{\textsf{terminate}\xspace}
\def\attac{\textsf{attach}\xspace}
\def\detac{\textsf{detach}\xspace}
\def\spli{\textsf{split}\xspace}
\def\merg{\textsf{merge}\xspace}

\newcommand{\piru}{\pi_{\mathsf{ru}}}
\newcommand{\piur}{\pi_{\mathsf{ur}}}
\newcommand{\piul}{\pi_{\mathsf{ul}}}
\newcommand{\pilu}{\pi_{\mathsf{lu}}}
\newcommand{\pild}{\pi_{\mathsf{ld}}}
\newcommand{\pidl}{\pi_{\mathsf{dl}}}
\newcommand{\pidr}{\pi_{\mathsf{dr}}}
\newcommand{\pird}{\pi_{\mathsf{rd}}}

\newcommand{\med}{\ensuremath{\delta}}
\newcommand{\divp}{\ensuremath{\delta}}
\newcommand{\lchild}{\ell_c}
\newcommand{\rchild}{r_c}
\newcommand{\segtree}{\mathcal{T}_\mathsf{seg}}
\newcommand{\wset}{\mathcal{W}}
\newcommand{\vset}{\mathcal{V}}

\newcommand{\bbox}{B}
\newcommand{\ch}{\mathsf{CH}}

\newbox\ProofSym
\setbox\ProofSym=\hbox{%
	\unitlength=0.18ex%
	\begin{picture}(10,10)
	\put(0,0){\framebox(9,9){}}
	\put(0,3){\framebox(6,6){}}
	\end{picture}}

\title{Minimum-Link Shortest Paths for Polygons amidst Rectilinear Obstacles\thanks{This research was partly supported by the Institute of Information \& communications 
Technology Planning \& Evaluation(IITP) grant funded by the Korea government(MSIT)
(No. 2017-0-00905, Software Star Lab (Optimal Data Structure and Algorithmic Applications in Dynamic Geometric Environment)) and (No. 2019-0-01906, Artificial Intelligence Graduate School Program(POSTECH)).}}
\author{Mincheol Kim\thanks{Department of Computer Science and Engineering, 
Pohang University of Science and Technology, Pohang, Korea. 
\texttt{rucatia@postech.ac.kr}}
\and Hee-Kap Ahn\thanks{Graduate School of Artificial Intelligence, Department of Computer Science and Engineering, 
Pohang University of Science and Technology, Pohang, Korea. 
\texttt{heekap@postech.ac.kr}}
}

\begin{document}

\maketitle
\begin{abstract}
Consider two axis-aligned rectilinear simple polygons in the domain 
consisting of axis-aligned rectilinear obstacles in the plane such that 
the bounding boxes, one for each obstacle and one for each polygon, 
are disjoint.
We present an algorithm that computes a minimum-link rectilinear shortest path 
connecting the two polygons in $O((N+n)\log (N+n))$ time using $O(N+n)$ space, 
where $n$ is the number of vertices in the domain 
and $N$ is the total number of vertices of the two polygons.
\end{abstract}


\section{Introduction}
The problem of finding paths connecting two objects 
amidst obstacles has been studied extensively in the past.
It varies on the underlying metric (Euclidean, rectilinear, etc.), types of obstacles (simple polygons,
rectilinear polygons, rectangles, etc.), and objective functions (minimum length, minimum number of links, 
or their combinations).
See the survey in Chapter 31 of~\cite{toth2017} on a wide variety of approaches to this problem and results.

For two points $p$ and $q$ contained in the plane, possibly with rectilinear polygonal obstacles 
(i.e., a \emph{rectilinear domain}), 
a rectilinear shortest path from $p$ to $q$ is a rectilinear path from $p$ to $q$ with minimum total length
that avoids the obstacles.
In the rest of the paper, we say a \emph{shortest path} to refer to 
a rectilinear shortest path unless stated otherwise.
A rectilinear path consists of horizontal and vertical segments, 
each of which is called \emph{link}.
Among all shortest paths from $p$ to $q$, we are interested in a \emph{minimum-link} 
shortest path from $p$ to $q$, that is, a shortest path with the minimum number of links
(or one with the minimum number of bends). There has been a fair amount of work on finding 
minimum-link shortest paths connecting two points amidst rectilinear obstacles in the plane~\cite{chen2001,lee1996,wang2019,yang1992,yang1995}.

These definitions are naturally extended to two more general objects contained in the domain.
A shortest path connecting the objects is one with minimum path length
among all shortest paths from a point of one object to a point of the other object.
A minimum-link shortest path connecting the objects is a minimum-link path
among all shortest paths. 

In this paper, we consider the problem of finding minimum-link shortest paths connecting two objects
in a rectilinear domain, which generalizes the case of connecting two points, in some modest environment.
The rectilinear polygonal obstacles are considered as open sets.
Two axis-aligned rectilinear polygons are said to be \emph{box-disjoint}
if the axis-aligned bounding boxes, one for each rectilinear polygon, are disjoint
in their interiors.
A set of axis-aligned rectilinear polygons is box-disjoint if the polygons of the set 
are pairwise box-disjoint.
The rectilinear domain induced by a set of box-disjoint rectilinear polygons
in the plane is called a \emph{box-disjoint rectilinear domain.}
We require the input objects and the obstacles in the domain
to be also pairwise box-disjoint, unless stated otherwise.

\paragraph{Problem definition.} Given two axis-aligned rectilinear simple polygons 
$\textsf{S}$ and $\textsf{T}$ in a rectilinear domain in the plane such that $\textsf{S}, \textsf{T}$, and 
the obstacles in the domain are pairwise box-disjoint, 
find a minimum-link rectilinear shortest path from $\textsf{S}$ to $\textsf{T}$.

\paragraph{Related work.}
Computing shortest paths or minimum-link paths in a polygonal domain
has been studied extensively.
When obstacles are all rectangles, 
Rezende~et~al.~\cite{rezende1989} presented an algorithm with $O(n \log n)$ time and $O(n)$ space
to compute a shortest path connecting two points amidst $n$ rectangles.
For a rectilinear domain with $n$ vertices, 
Mitchell~\cite{mitchell1992} gave an algorithm with $O(n\log n)$ time 
and $O(n)$ space to compute a shortest path connecting two points
using a method based on the continuous Dijkstra paradigm~\cite{mitchell1989}.
Later, Chen and Wang~\cite{chen2013} improved the time complexity to
$O(n + h \log h)$ for a triangulated polygonal domain with $h$ holes. 

Computing a minimum-link path, not necessarily shortest, 
in a polygonal domain have also been studied well.
For a minimum-link rectilinear path connecting two points in a rectilinear domain with $n$ vertices,
Imai and Asano~\cite{imai1986} gave an algorithm with $O(n \log n)$ time and space.
Then a few algorithms improved the space complexity to $O(n)$ 
without increasing the running time~\cite{das1991,mitchell2014,sato1987}.
Very recently, Mitchell~et~al.~\cite{mitchell2019} gave an algorithm with $O(n + h\log h)$ time 
and $O(n)$ space for triangulated rectilinear domains with $h$ holes.

Yang~et~al.~\cite{yang1992} considered the problem of finding a 
rectilinear path connecting two points amidst rectilinear obstacles 
under a few optimization criteria, such as 
a minimum-link shortest path, a shortest minimum-link path, 
and a least-cost path (a combination of link cost and length cost).
By constructing a path-preserving graph, they gave a unified approach to compute such paths
in $O(ne+n\log n)$ time, where $n$ is the total number of polygon edges and 
$e$ is the number of polygon edges connecting two convex vertices.
The space complexity is $O(ne)$ due to the path-preserving graph of size $O(ne)$.
Since $e$ is $O(n)$,
the running time becomes $O(n^2)$ in the worst case, even for 
convex rectilinear polygons (obstacles). 
A few years later, they gave two algorithms on the problem~\cite{yang1995}, 
improving their previous result, 
one with $O(n\log^2 n)$ time and $O(n\log n)$ space and the other with $O(n\log^{3/2} n)$ time 
and $O(n\log^{3/2} n)$ space using 
a combination of a graph-based approach and the continuous Dijkstra approach.
It is claimed in~\cite{lee1996} that a minimum-link shortest path 
can be computed in $\Theta(n \log n)$ time and $O(n)$ space
when obstacles are rectangles by the algorithm in another paper~\cite{Yang1992rectangle} by
the same authors, but 
the paper~\cite{Yang1992rectangle} is not available.

Later, Chen~et~al.~\cite{chen2001} gave an algorithm improving the previous results 
for finding a minimum-link shortest path connecting two points in $O(n\log^{3/2} n)$ time and $O(n\log n)$ 
space using an implicit representation of a reduced visibility graph, 
instead of computing the whole graph explicitly.
Very recently, Wang~\cite{wang2019} pointed out a flaw in the algorithm in~\cite{chen2001} 
and claimed that to make it work, each vertex of the graph
must store a constant number of nonlocal optimum paths together with local optimum paths.
Wang gave an algorithm with $O(n + h\log^{3/2}h)$ time and $O(n + h\log h)$ space
using a reduced path-preserving graph from the corridor structure~\cite{mitchell2019} and the histogram partitions~\cite{schuierer1996},
where $h$ is the number of holes (obstacles) in the rectilinear domain.

However, we are not aware of any result on computing the minimum-link shortest path 
connecting two objects other than points.

\paragraph{Our results.}
We consider the minimum-link shortest path problem for two axis-aligned rectilinear polygons
$\textsf{S}$ and $\textsf{T}$ 
in a box-disjoint rectilinear domain. This generalizes the two-point shortest path problem
to two-polygon shortest path problem.
The theorem below summarizes our results.

\begin{theorem}\label{thm:result}
Let $\textsf{S}$ and $\textsf{T}$ be two axis-aligned rectilinear simple polygons  
with $N$ vertices in a rectilinear domain with $n$ vertices in the plane such that
$\textsf{S}, \textsf{T}$, and the obstacles in the domain are pairwise box-disjoint.
We can compute a minimum-link shortest path from $\textsf{S}$ to $\textsf{T}$
in $O((N+n)\log(N+n))$ time using $O(N+n)$ space.
\end{theorem}

\paragraph{Sketch of our results.}
The main difficulty lies in computing a shortest path from $\textsf{S}$ to $\textsf{T}$.
The length of a shortest path from $\textsf{S}$ to $\textsf{T}$ is determined by
a pair of points, one lying on the boundary of $\textsf{S}$ and one lying on the
boundary of $\textsf{T}$. Such a point is a vertex of $\textsf{S}$ or $\textsf{T}$,
or the first intersection of a horizontal or vertical ray emanating
from a vertex in the domain with the boundary of $\textsf{S}$ or $\textsf{T}$. 
Since the domain has $O(N+n)$ vertices,
there are $O(N+n)$ such points on the boundaries of $\textsf{S}$ and $\textsf{T}$,
and $O((N+n)^2)$ pairs of points, one from $\textsf{S}$ and the other from $\textsf{T}$, 
to consider in order to determine the length of a shortest path. 
Thus, 
if we use a naive approach that computes a minimum-link shortest path for each point pair, 
it may take $\Omega((N+n)^2)$ time.
Theorem~\ref{thm:result} shows that our algorithm computes 
a minimum-link shortest path from $\textsf{S}$ to $\textsf{T}$ efficiently.
Also, a minimum-link shortest path can intersect the interiors of bounding boxes of obstacles,
although $\textsf{S}$, $\textsf{T}$, and obstacles are pairwise box-disjoint.
See Figure~\ref{fig:boxdisjoint}.

\begin{figure}[t]
  \begin{center}
    \includegraphics[width=0.77\textwidth]{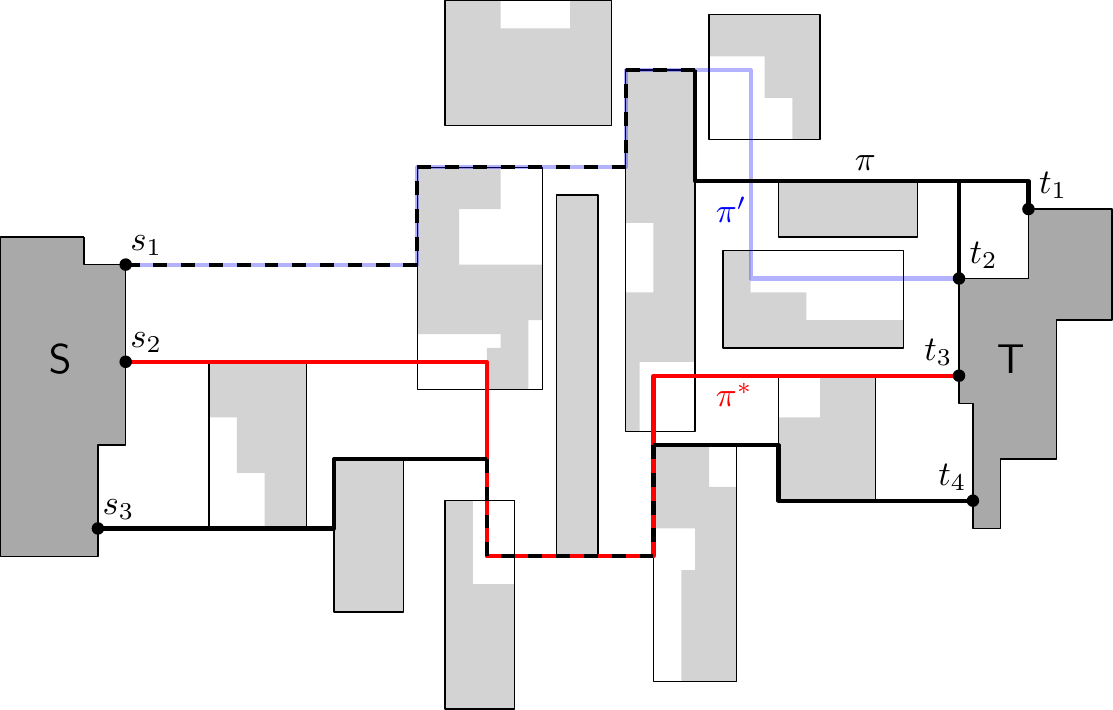}
    \caption{\small
The box-disjoint rectilinear domain with $\textsf{S}$, $\textsf{T}$. Light gray rectilinear polygons are obstacles.
There are six pairs of points,
$(s_1,t_1),(s_1,t_2),(s_2,t_3),(s_2,t_4),(s_3,t_3),(s_3,t_4)$,
that determine the length of a shortest path from $\textsf{S}$ to $\textsf{T}$.
The path $\pi$ from $s_1$ to $t_1$ (or $t_2$) is a minimum-link shortest path from $\textsf{S}$ to $\textsf{T}$
with eight links among all shortest paths without intersecting the interiors of bounding boxes of obstacles.
However,
the blue path $\pi'$ from $s_1$ to $t_2$ has seven links and
the red path $\pi^*$ from $s_2$ to $t_3$ has five links, which is optimal.}
    \label{fig:boxdisjoint}
  \end{center}
\end{figure}

We first consider a simpler problem
for an axis-aligned line segment $S$ and a point $t$ contained in the domain consisting of
axis-aligned rectangular obstacles.
We partition the domain into at most eight regions using eight $xy$-monotone
paths from $S$. We observe that every shortest path from $S$ to a point
in a region is either $x$-, $y$-, or $xy$-monotone~\cite{rezende1989}. 
Moreover, we define a set of $O(n)$ baselines for each region, and
show that there is a minimum-link shortest path from $S$ to $t$ consisting of segments contained in the baselines. 
Based on these observations, our algorithm applies a plane sweep technique with a sweep line moving
from $S$ to $t$ and computes the minimum numbers of links at the intersections of the baselines 
and the sweep line efficiently. After the sweep line reaches $t$, our algorithm reports
a minimum-link shortest path that can be obtained from a reverse traversal 
from $t$ using the number of links stored in baselines. 
During the sweep, our algorithm maintains a data structure storing baselines
(and their minimum numbers of links) 
and updates the structure for the segments (events)
on the boundary of the region.

It takes, however, $O(n^2)$ time using $O(n)$ space.
To reduce the time complexity without increasing the space complexity, 
our algorithm maintains another data structure, a balanced binary search tree, each
node of which corresponds to a set of consecutive baselines. 
This tree behaves like a segment tree~\cite{de2008}.
Instead of updating the minimum numbers of links of $O(n)$ baselines at each event of the plane sweep
algorithm, we update $O(\log n)$ nodes of the tree that together correspond to the baselines.
This improves the time for handling each sweep-line event from $O(n)$ to $O(\log n)$,
and thus improving the total time complexity to $O(n\log n)$.
This tree can be constructed and maintained using $O(n)$ space, and thus the total space remains
to be $O(n)$.

Then we extend our algorithm to handle a line segment $T$ (not a point $t$) 
and box-disjoint rectilinear obstacles (not necessarily rectangles).
We observe that every shortest path contained in a region from $S$ to any point of $T$ 
is either $x$-, $y$-, or $xy$-monotone, so our algorithm partitions the domain into at most eight regions again.
Then $T$ intersects at most five regions. 
Our algorithm computes a minimum-link shortest path from $S$ to $T'$ for the portion $T'$ of $T$ 
contained in each region, and then returns the minimum-link shortest path among the paths.

When $S$ or $T$ intersects some bounding boxes of obstacles,
we consider each portion of $S$ or $T$ contained in a bounding box independently.
The portion not contained in any bounding box can be handled as we do for segments disjoint from the boxes. 
For the portion contained in a bounding box $\bbox(P)$ for a rectilinear polygon $P$,
every minimum-link shortest path from $S$ to $T$ is the concatenation of
a subpath contained in $\bbox(P)$ and the subpath not contained in $\bbox(P)$
such that both subpaths are minimum-link shortest paths sharing one point on the boundary 
of $\bbox(P)$. Thus, our algorithm finds 
a subpath contained in $\bbox(P)$ and a subpath not contained 
in $\bbox(P)$ that together form a minimum-link shortest
path from $S$ to $T$. We observe that there is a minimum-link shortest path from $S$ to $T$
through certain points on the boundary of $\bbox(P)$. By computing these
points and their distances and minimum number of links to $S$ and $T$, our algorithm computes
a minimum-link shortest path from $S$ to $T$ in $O(n\log n)$ time for $S$ or $T$ intersecting 
the bounding boxes. Since an axis-aligned line segment intersects at most two bounding boxes, 
the overall running time remains to be $O(n\log n)$ time using $O(n)$ space.

Finally, we consider that the input objects are rectilinear
simple polygons $\mathsf{S}$ and $\mathsf{T}$ with $N$ vertices.
Recall that there are $O((N+n)^2)$ pairs of points that 
determine the length of a shortest path from $\textsf{S}$ to $\textsf{T}$.
To handle them efficiently, we add $O(N)$ additional baselines and
$O(N)$ events induced by $\textsf{S}$ and $\textsf{T}$ during the plane sweep algorithm.
Then the number of events becomes $O(N+n)$
and the time to handle each event takes $O(\log(N+n))$,
so we obtain Theorem~\ref{thm:result}.

\section{Preliminaries}
Let $\rectset$ be a set of $n$ disjoint axis-aligned rectangles in $\realp$.
Each rectangle $R\in\rectset$ is considered as an open set
and plays as an obstacle in computing a minimum-link shortest path in the plane. 
We let $\dom := \realp - \cup_{R\in\rectset}R$ and call it the \emph{rectangular domain}
induced by $\rectset$ in the plane. 
For two points $p$ and $q$ in $\dom$,
$d(p,q)$ denotes the $L_1$ distance (or the Manhattan distance) from $p$ to $q$ in $\dom$, that is,
the length of a shortest path from $p$ to $q$ avoiding the obstacles.
A path is \emph{$x$-monotone} if the intersection of the path with any line
perpendicular to the $x$-axis is connected.
Likewise, a path is \emph{$y$-monotone} if the intersection of the path
with any line perpendicular to the $y$-axis is connected.
If a path is $x$-monotone and $y$-monotone, the path is \emph{$xy$-monotone}.

For two objects $\textsf{S}$ and $\textsf{T}$ in $\dom$, 
$d(\textsf{S},\textsf{T})=\min_{p\in\textsf{S}, q\in\textsf{T}}d(p,q)$.
A shortest path from $\textsf{S}$ to $\textsf{T}$ is a path in $\dom$ from a point 
$p\in \textsf{S}$ to a point $q\in \textsf{T}$ of length $d(\textsf{S},\textsf{T})$.
A \emph{minimum-link shortest path} from $\textsf{S}$ to $\textsf{T}$ is a path 
that has the minimum number of links among all shortest paths from $\textsf{S}$ to $\textsf{T}$ in $\dom$,
and we use $\ml(\textsf{S},\textsf{T})$ to denote the number of links of 
a minimum-link shortest path from $\textsf{S}$ to $\textsf{T}$.
We call a pair $(p,q)$ of points with $p\in\textsf{S}$ and $q\in\textsf{T}$ such that $d(\textsf{S},\textsf{T}) = d(p,q)$
a \emph{closest pair} of points of $\textsf{S}$ and $\textsf{T}$.
We say $p$ is a closest point of $\textsf{S}$ from $\textsf{T}$, and $q$ is a closest point of $\textsf{T}$ from $\textsf{S}$.
Note that there can be more than one closest pair of points of $\textsf{S}$ and $\textsf{T}$.

We make an assumption that the rectangles are in \emph{general position},
that is, no two rectangles in $\rectset$ have corners, one corner from each rectangle, 
with the same $x$- or $y$-coordinate.
A horizontal line segment $H$ can be represented by
the two $x$-coordinates $x_1(H)$ and $x_2(H)$ of its endpoints ($x_1(H) < x_2(H)$)
and the $y$-coordinate $y(H)$ of them.
Likewise, a vertical line segment $V$ can be represented by
the two $y$-coordinates $y_1(V)$ and $y_2(V)$ of its endpoints ($y_1(V) < y_2(V)$)
and the $x$-coordinate $x(V)$ of them.

\subsection{Eight disjoint regions of a rectangular domain}\label{sec:eightpath}
Given a rectangular domain $\dom$ and a vertical segment $S$, 
we partition $\dom$ into at most eight disjoint regions
by using eight $xy$-monotone paths from the endpoints of $S$
in a way similar to the one by Choi and Yap~\cite{choi1996}.
Consider a horizontal ray from a point $p=p_1$ on $S$ going rightwards.
The ray stops when it hits a rectangle $R\in\rectset$ at a point $p_1'$. 
Let $p_2$ be the top-left corner of
$R$. We repeat this process by taking a horizontal ray from
$p_2$ going rightwards until it hits a rectangle, and so on. 
The last horizontal ray goes to infinity.
Then we obtain an $xy$-monotone path
$\piru(p)=(p=p_1p_1'p_2p_2'\ldots)$.
In other words, $\piru(p)$ is an $xy$-monotone path from $p$
that alternates going \emph{rightwards} (until hitting a rectangle) and 
going \emph{upwards} (to the top-left corner of the rectangle).

By choosing two directions, one going either rightwards or leftwards horizontally,
and one going either upwards or downwards vertically, and ordering the chosen directions, 
we define eight rectilinear $xy$-monotone paths with directions:
rightwards-upwards (\textsf{ru}), upwards-rightwards (\textsf{ur}), upwards-leftwards (\textsf{ul}), leftwards-upwards (\textsf{lu}), 
leftwards-downwards (\textsf{ld}), downwards-leftwards (\textsf{dl}), downwards-rightwards (\textsf{dr}),
and rightwards-downwards (\textsf{rd}).
We use $\pi_\alpha(p)$ to denote them, where $\alpha$ is one in 
$\{\mathsf{ru}, \mathsf{ur}, \mathsf{ul}, \mathsf{lu}, \mathsf{ld}, \mathsf{dl}, \mathsf{dr}, \mathsf{rd}\}$.
Also, we use $\pi_\alpha(p,q)$ to denote the subpath of $\pi_\alpha(p)$ from $p$ to $q\in\pi_\alpha(p)$.

\begin{figure}[t]
  \begin{center}
    \includegraphics[width=0.6\textwidth]{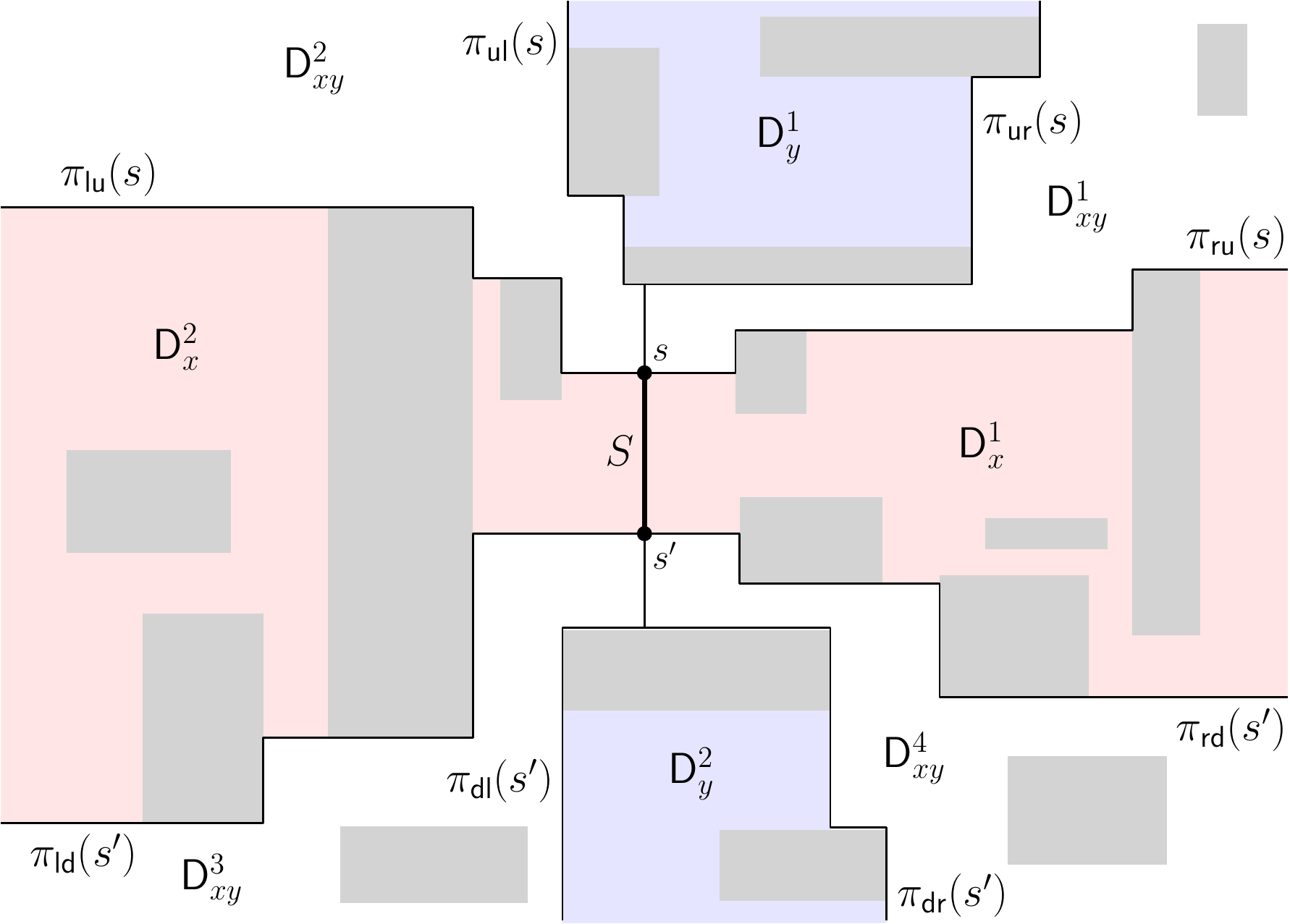}
    \caption{\small
      Eight disjoint regions of $\dom$ by eight $xy$-monotone paths from $s$ or $s'$. Gray rectangles are obstacles.
    }
    \label{fig:subregion}
  \end{center}
\end{figure}

Figure~\ref{fig:subregion} illustrates these eight $xy$-monotone paths, 
four upward paths from
the upper endpoint $s$ of $S$ and four downward paths from the lower endpoint $s'$ of $S$.
Observe that for a point $p\in\dom$, the eight paths $\pi_\alpha(p)$ do not cross each other,
so the four upwards paths from $s$ and the four downwards paths from $s'$
do not cross each other.
Thus, by the eight paths, $\dom$ is partitioned into eight regions. See Figure~\ref{fig:subregion}.
We denote by $\xydom^1$ (and $\xydom^2$, $\xydom^3$, $\xydom^4$)
the region bounded by $\piru(s)$ and $\piur(s)$
(and by $\piul(s)$ and $\pilu(s)$, by $\pild(s')$ and $\pidl(s')$, by $\pidr(s')$ and $\pird(s')$).
We denote by $\xdom^1$ (and $\xdom^2$) the region bounded by $\piru(s)$ and $\pird(s')$
(and by $\pilu(s)$ and $\pild(s')$), and denote by
$\ydom^1$ (and $\ydom^2$) the region bounded by $\piur(s)$ and $\piul(s)$
(and by $\pidl(s')$ and $\pidr(s')$). 

\begin{lemma}\label{lem:monotone}
For a point $t\in\cup_{1\leq i \leq 4}\xydom^i$,
every shortest path from $S$ to $t$ is $xy$-monotone.
For a point $t\in\cup_{1\leq i \leq 2}\xdom^i$,
every shortest path from $S$ to $t$ is $x$-monotone.
For a point $t\in\cup_{1\leq i \leq 2}\ydom^i$,
every shortest path from $S$ to $t$ is $y$-monotone.
\end{lemma}
\begin{proof}
We claim that every shortest path from $S$ to $t$ connects 
the upper endpoint $s$ of $S$ and $t$
for a point $t\in\ydom^1\cup\xydom^1\cup\xydom^2$.
Assume to the contrary that a shortest path $\pi$ 
from $S$ to $t$ does not pass through $s$.
Then $\pi$ crosses $\piru(s)$ (or $\pilu(s)$) at a point $t'$. By replacing the
portion of $\pi$ from $S$ to $t'$ with the portion of $\piru(s)$ from $s$
to $t'$, we can get a shorter path, a contradiction. 
By a similar argument, we observe that every shortest path from $S$ to $t$ connects
the lower endpoint $s'$ of $S$ and $t$ if $t\in\ydom^2\cup\xydom^3\cup\xydom^4$.

Rezende et al.~\cite{rezende1989} showed that
every shortest path connecting two points in $\dom$ is 
$x$-, $y$-, or $xy$-monotone. Choi and Yap~\cite{choi1996}
gave a classification that for a point $t\in\cup_{1\leq i \leq 2}\ydom^i$
every shortest path from $S$ to $t$ is $y$-monotone,
and for a point $t\in\cup_{1\leq i \leq 4}\xydom^i$ every shortest path from $S$ to $t$ 
is $xy$-monotone.
Assume that $t\in\xdom^1$.
Both $\piul(t)$ and $\piur(t)$ intersect $\piru(s)$,
and both $\pidl(t)$ and $\pidr(t)$ intersect $\pird(s')$.
This implies that every shortest path from a point in $S$ to $t$ 
is $x$- or $xy$-monotone by the classification of Choi and Yap~\cite{choi1996}.
Hence every shortest path from $S$ to $t$ is $x$-monotone.
The case for $p\in\xdom^2$ can be shown similarly.
\end{proof}

From now on we simply use $\xydom$, $\xdom$ and $\ydom$ to denote
$\xydom^1$, $\xdom^1$ and $\ydom^1$, respectively, and 
assume that $t$ lies in a region $\dom'$ of the regions.
The case that $t$ lies in other regions can be handled analogously.
For each horizontal side of the rectangles incident to $\dom'$,
we call the horizontal line containing the side a \emph{horizontal baseline} of $\dom'$.
Similarly, for each vertical side of the rectangles incident to $\dom'$,
we call the vertical line containing the side a \emph{vertical baseline} of $\dom'$.
The two vertical lines through $S$ and $t$, and 
the three horizontal lines through $s$, $s'$ and $t$ are 
also regarded as vertical and horizontal baselines of $\dom'$, respectively.
We say a minimum-link shortest path $\pi$ 
is \emph{aligned to the baselines} if every segment of $\pi$ is contained in a baseline
of the corresponding region. 
By using Lemma~\ref{lem:base-align}, we find a minimum-link shortest path aligned to
the baselines of each region.

\begin{lemma}\label{lem:base-align}
There is a minimum-link shortest path from $S$ to $t$ that is aligned to the baselines
of $\dom'$.
\end{lemma}
\begin{proof}
Assume that a minimum-link shortest path $\pi$ has a horizontal line segment $H$
which is not contained in any baseline of $\dom'$. 
Clearly, $H$ is not incident to $t$, because there is a horizontal baseline
through $t$.
If $H$ is incident to $S$ not at its endpoints, 
we can move $H$ vertically and get a shorter path, a contradiction.
If both vertical segments of $\pi$ incident to $H$ are 
contained in one side of the line through $H$, then we can get a path shorter than 
$\pi$ by moving $H$ towards the side and shortening the two vertical segments
incident to $H$, a contradiction. This also applies to a vertical segment of $\pi$ not contained
in any baseline. 

Now assume that $\pi$ is $xy$-monotone and $H$ is incident to neither $S$ nor $t$.
Let $\pi'$ be a maximal subpath of $\pi$ such that $\pi'$ contains $H$,
and no horizontal baseline of $\dom'$ intersects $\pi'$ except at its two endpoints $p_1$ and $p_2$
with $y(p_1) < y(p_2)$. See Figure~\ref{fig:lemma3}(a).
We show that the axis-aligned rectangle $R$ with corners at $p_1$ and $p_2$, is contained in $\dom'$.
Assume to the contrary that $R$ is not contained in $\dom'$,
that is, there is a rectangle $R'\in\rectset$ incident to $\dom'$ that intersects $R$.
Then there is a horizontal baseline of $\dom'$ through a side of $R'$
that intersects $\pi'$. This contradicts the definition of $\pi'$,
so $R$ is contained in $\dom'$.
See Figure~\ref{fig:lemma3}(a).

\begin{figure}[t]
  \begin{center}
    \includegraphics[width=0.8\textwidth]{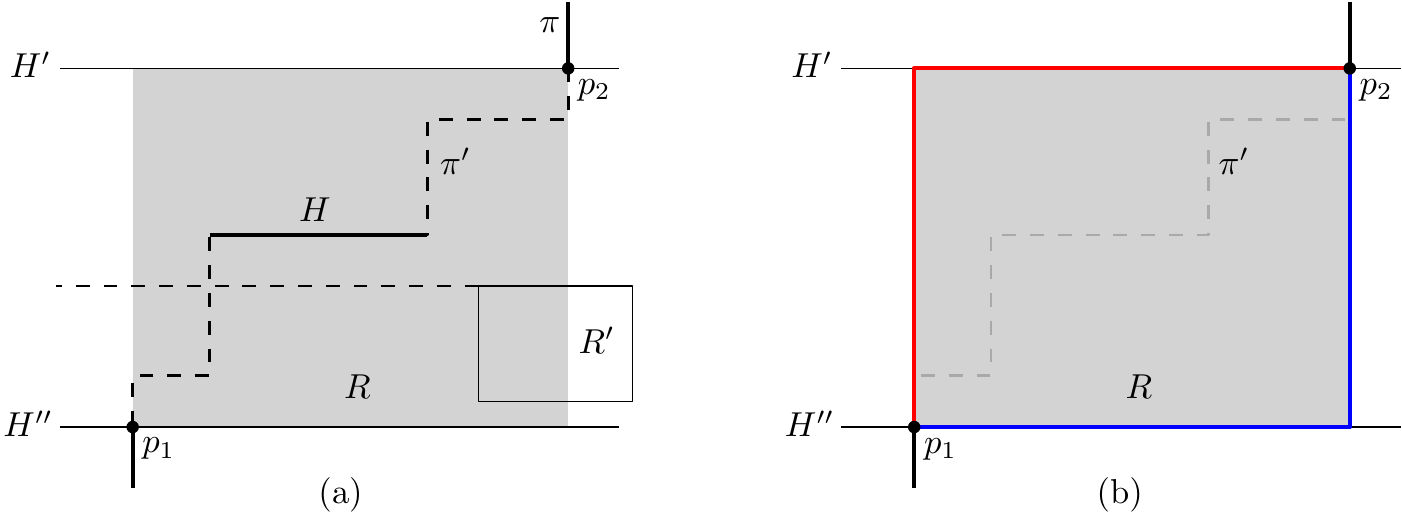}
    \caption{\small
Proof of Lemma~\ref{lem:base-align}.
(a) When $\pi'$ intersects two baselines $H'$ and $H''$ at its two endpoints $p_1, p_2$,
no rectangle intersects the rectangle $R$ (gray) with corners at $p_1, p_2$.
(b) By replacing $\pi'$ with the subpath (red or blue) along the boundary of $R$, 
we obtain a path with a smaller or the same number of links.
    }
    \label{fig:lemma3}
  \end{center}
\end{figure}

Thus, we can replace the subpath $\pi'$ with a horizontal side and a vertical side of $R$
without increasing the length of $\pi$.
See Figure~\ref{fig:lemma3}(b).
The resulting path has the number of links smaller than or equal to that of $\pi$.
By applying the procedure above for every horizontal line segment 
not contained in a horizontal baseline of $\dom'$, we can get
a minimum-link shortest path $\pi^*$ from $s^*$ to $t$ such that
every horizontal line segment of $\pi^*$ is contained in a horizontal baseline of $\dom'$.

Similarly, we can replace every vertical segment of $\pi'$ not contained in a vertical baseline
with one contained in a baseline without increasing the length of the path.
\end{proof}

\section{\texorpdfstring{$t$}{t} lies in \texorpdfstring{$\xydom$}{Dxy}}
\label{sec:xymonotone}
We consider the case that $t$ lies in $\xydom$.
By Lemma~\ref{lem:monotone}, every shortest path from $S$ to $t$ is $xy$-monotone and 
connects the upper endpoint $s$ of $S$ and $t$.
Let $c$ be the point with the maximum 
$x$-coordinate and the maximum $y$-coordinate
among the points in $\piur(s)\cap\pild(t)$. Observe that
$c$ is defined uniquely as $\piur(s)\cap\pild(t)$ is connected
and $xy$-monotone by the definition.
Likewise, let $c'$ be the point with the maximum 
$x$-coordinate and the maximum $y$-coordinate
among the points in $\piru(s)\cap\pidl(t)$.
Then we use $\xydom(s,t)$ to denote the region of $\xydom$
enclosed by the closed curve composed of $\piur(s,c)$, $\pild(t,c)$, $\piru(s,c')$, and $\pidl(t,c')$.
We denote by $\xybd(s,t)$ the rectilinear chain of the outer boundary of $\xydom(s,t)$
from $s$ to $t$ in clockwise order,
and denote by $\xybd(t,s)$ the rectilinear chain of the outer boundary of $\xydom(s,t)$
from $t$ to $s$ in clockwise order.
See Figure~\ref{fig:xydom}(a) for an illustration.
By Lemma~\ref{lem:monotone}, every shortest path from $s$ to $t$ is contained in 
$\xydom(s,t)$, and therefore every minimum-link shortest path from $s$ to $t$
is also contained in $\xydom(s,t)$.

\begin{figure}[t]
  \begin{center}
    \includegraphics[width=\textwidth]{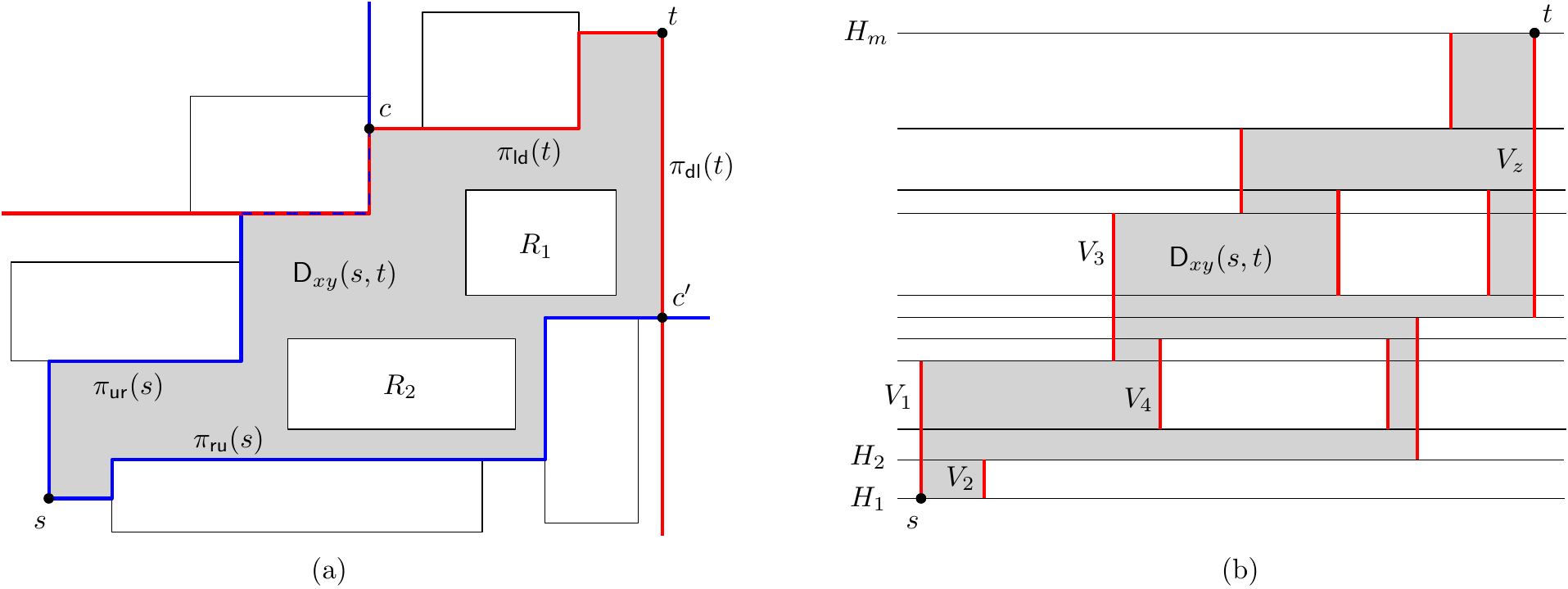}
    \caption{\small (a) $\xydom(s,t)$ is the region of $\xydom$ enclosed by the closed curve composed of
    $\piur(s,c)$, $\pild(t,c)$, $\piru(s,c')$, and $\pidl(t,c')$.
$R_1$ and $R_2$ are the holes of $\xydom(s,t)$.
(b) Horizontal baselines $H_1, H_2,\ldots, H_m$ of $\xydom(s,t)$ and vertical segments (red) $V_1, V_2,\ldots, V_z$ on the boundary of $\xydom(s,t)$.
    }
    \label{fig:xydom}
  \end{center}
\end{figure}

We focus on the baselines of $\xydom$ that are 
defined by $s$, $t$, and the rectangles incident to $\xydom(s,t)$, which 
we call the baselines of $\xydom(s,t)$.
Figure~\ref{fig:xydom}(b) shows the horizontal baselines of $\xydom(s,t)$.
Note that a baseline
may cross rectangles incident to $\xydom(s,t)$.
Let $H_1, H_2, \ldots, H_m$ be the $m$ horizontal baselines of $\xydom(s,t)$ such that
$y(H_1) < y(H_2) < \ldots < y(H_m)$.
Note that $s$ is on $H_1$ and $t$ is on $H_m$.

\subsection{Computing the minimum number of links}\label{sec:min-link}
Consider a minimum-link shortest path aligned to the baselines of $\xydom(s,t)$.
For the rightmost vertical segment $V$ of $\xydom(s,t)$, we have $y_2(V) = y(t)$
and $y_1(V) = y(H_{m'})$ for some horizontal baseline $H_{m'}$ with $m'<m$.
We can compute a minimum-link shortest path
once we have a minimum-link shortest path from $s$ 
to the intersection point $c_i$ of $V$ and $H_i$ for each $i=m',m'+1,\ldots,m$,
since $t$ is the endpoint of $V$.

We compute $\ml(s,t)$ by applying the plane sweep algorithm,
and then report a minimum-link shortest path aligned to the baselines of $\xydom(s,t)$
that can be obtained from a reverse traversal from $t$ using $\ml(s,t)$.

Imagine a vertical line $L$ sweeping $\xydom(s,t)$ rightwards.
Our plane sweep algorithm maintains a data structure storing horizontal baselines
and their minimum numbers of links among shortest paths from $s$ to intersections of baselines and $L$
such that the line segments incident to the intersections of those shortest paths are horizontal.
The algorithm updates their status and minimum numbers of links
when $L$ encounters the vertical segments (vertical baselines) on the boundary of $\xydom(s,t)$.

We define the status for each horizontal baseline as follows.
For the intersection point $c_i=H_i\cap L$ for each $i=1,\ldots,m$,
if $c_i \in \xydom(s,t)$, then $H_i$ is \emph{active}. Otherwise, $H_i$ is 
\emph{inactive}. Observe that a baseline may switch its status between active and inactive,
depending on the position of $L$, and these switches occur only when $L$ encounters 
a vertical segment on the boundary of $\xydom(s,t)$.
During the sweep, we maintain the active baselines of $\xydom(s,t)$ in a set of ranges 
with respect to their indices in a range tree $\rangebbst$.
A range $[a,b]$ contained in $\rangebbst$ represents a set of active baselines
$H_a, H_{a+1},\ldots,H_b$, consecutive in their indices from $a$ to $b$.
Every range $[a,b]$ in $\rangebbst$ is \emph{maximal} 
in the sense that $H_{a-1}$ and $H_{b+1}$ are inactive or not defined in $\xydom(s,t)$.
We use $M(i)$ to denote the minimum number of links among all shortest paths from $s$ to $c_i$
whose segment incident to $c_i$ is horizontal.

We maintain $M(i)$'s for horizontal baselines during the plane sweep as follows.
There are vertical line segments $V_1, V_2, \ldots, V_z$ on the boundary of $\xydom(s,t)$,
satisfying $x(V_1) < x(V_2) < \ldots < x(V_z)$.
Note that the lower endpoint of $V_1$ is $s$ and the upper endpoint of $V_z$ is $t$.
We consider each vertical segment $V_j$ $(1\leq j\leq z)$ of $\xydom(s,t)$ as an event, 
denoted by $E_j$,
because we compute a minimum-link shortest path aligned to the baselines of $\xydom(s,t)$,
so $M(i)$ changes only when $L$ encounters a vertical segment.
For each $E_j$, we use $\alpha(j)$ and $\beta(j)$ (with $\alpha(j)<\beta(j)$)
to denote the indices such that $y_1(V_j) = y(H_{\alpha(j)})$ and $y_2(V_j) = y(H_{\beta(j)})$, respectively.
$E_j$ belongs to one of the following six types depending on the boundary part of $\xydom(s,t)$ 
that $V_j$ lies on. See Figure~\ref{fig:6events} for an illustration of each type.
\begin{itemize}
\item $E_1$ belongs to type {\origi} and $E_z$ belongs to type {\termi.}
\item $E_j$ for each $j=2,\ldots, z-1$ belongs to type {\attac} if $V_j$ lies on $\xybd(s,t)$,
and to type {\detac} if $V_j$ lies on $\xybd(t,s)$.
\item $E_j$ belongs to type {\spli}  if $V_j$ is 
the left side of a hole of $\xydom(s,t)$, and to type {\merg} if $V_j$ is 
the right side of a hole.
\end{itemize}
\begin{figure}[t]
  \begin{center}
    \includegraphics[width=\textwidth]{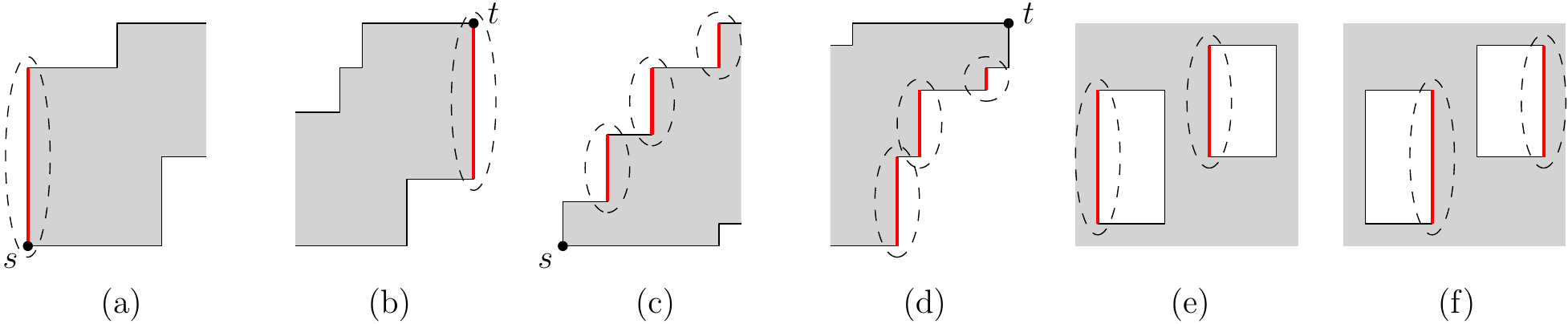}
    \caption{\small Six types of events of the plane sweep algorithm. (a) \origi event. (b) \termi event. 
(c) \attac event. (d) \detac event. (e) \spli event. (f) \merg event. 
    }
    \label{fig:6events}
  \end{center}
\end{figure}

The events are sorted by their $x$-coordinates.
During the sweep, $L$ encounters $E_j$ when $x(L) = x(V_j)$. 
Initially, the tree $\rangebbst$ contains no range, and $M(i)$ is set to $\infty$ for all horizontal 
baselines $H_i$.
When $L$ encounters $E_1$, which is the \origi event with $\alpha(1)=1$,
we update $M(\alpha(1)) := 1$ and $M(k) := 2$ for each $k\in[\alpha(1)+1,\beta(1)]$,
and insert the range $[\alpha(1),\beta(1)]$ into $\rangebbst$.

If $E_j$ is an \attac event, 
the inactive baselines $H_i$ for $i$ from $\alpha(j)+1$ to $\beta(j)$ become active.
Observe that there always exists a range $[a',\alpha(j)]$ in $\rangebbst$ with $a' < \alpha(j)$.
Thus, we remove $[a',\alpha(j)]$ from $\rangebbst$ and insert $[a',\beta(j)]$ into $\rangebbst$.
Then we update $M(i):=\min_{k\in[a',\alpha(j)]}\{M(k)+2\}$ for each $i\in[\alpha(j)+1,\beta(j)]$.

If $E_j$ is a \detac event, 
the active baselines $H_i$ for $i$ from $\alpha(j)$ to $\beta(j)-1$ become inactive.
Observe that there always exists a range $[\alpha(j),b']$ in $\rangebbst$ with $\beta(j) < b'$.
Thus, we remove $[\alpha(j),b']$ from $\rangebbst$, and insert $[\beta(j),b']$ into $\rangebbst$.
Then we update $M(i):=\min\{M(i),\min_{k\in[\alpha(j),\beta(j)-1]}(M(k)+2)\}$ for each $i\in[\beta(j),b']$.

If $E_j$ is a \spli event, 
the active baselines lying in between $H_{\alpha(j)}$ and $H_{\beta(j)}$
become inactive.
If there is such a baseline, there always exists a range $[a',b']$
in $\rangebbst$ with $a' < \alpha(j)$ and $\beta(j) < b'$.
In this case, we remove $[a',b']$ from $\rangebbst$, insert $[a',\alpha(j)]$ and $[\beta(j),b']$ into $\rangebbst$,
and update for each $i\in[\beta(j),b']$
$M(i) :=\min\{M(i),\min_{k\in[a',\beta(j)-1]}\{M(k)+2\}\}$ for each $i\in[\beta(j),b']$.

If $E_j$ is a \merg event,
the inactive baselines lying in between $H_{\alpha(j)}$ and $H_{\beta(j)}$
become active.
If there is such a baseline,
there always exist two ranges $[a',\alpha(j)]$ and $[\beta(j),b']$ 
in $\rangebbst$ with $a' < \alpha(j)$ and $\beta(j) < b'$.
In this case, we remove $[a',\alpha(j)]$ and $[\beta(j),b']$ from $\rangebbst$, insert $[a',b']$ into $\rangebbst$, and
update 
\begin{equation}\label{eqn:merge}
M(i):=\left.
\begin{cases}
\min_{k\in[a',\alpha(j)]}M(k)+2 & \text{ for }i\in[\alpha(j)+1,\beta(j)-1],\\
\min\{M(i),\min_{k\in[a',\alpha(j)]}M(k)+2\}& \text{ for }i\in[\beta(j),b'].
\end{cases}
\right.
\end{equation}

Our algorithm eventually finds $\ml(s,t)$ when 
$L$ encounters the \termi event $E_z$ with $\beta(z)=m$.
Then $\rangebbst$ has exactly one range $[\alpha(z),\beta(z)]$, and we remove it from $\rangebbst$.
We take $\ml(s,t)=\min\{M(m),\min_{k\in[\alpha(z),\beta(z)-1]}M(k)+1\}$.

\subsection{Computing a minimum-link shortest path}\label{sec:mlsp}
We compute a minimum-link shortest path from $s$ to $t$ aligned to the baselines of $\xydom(s,t)$ 
using $\ml(s,t)$.
To do this, we add a horizontal line segment at each event, 
which we call a \emph{canonical segment}.
Then we report a minimum-link shortest path using these canonical segments.

\begin{figure}[t]
  \begin{center}
    \includegraphics[width=\textwidth]{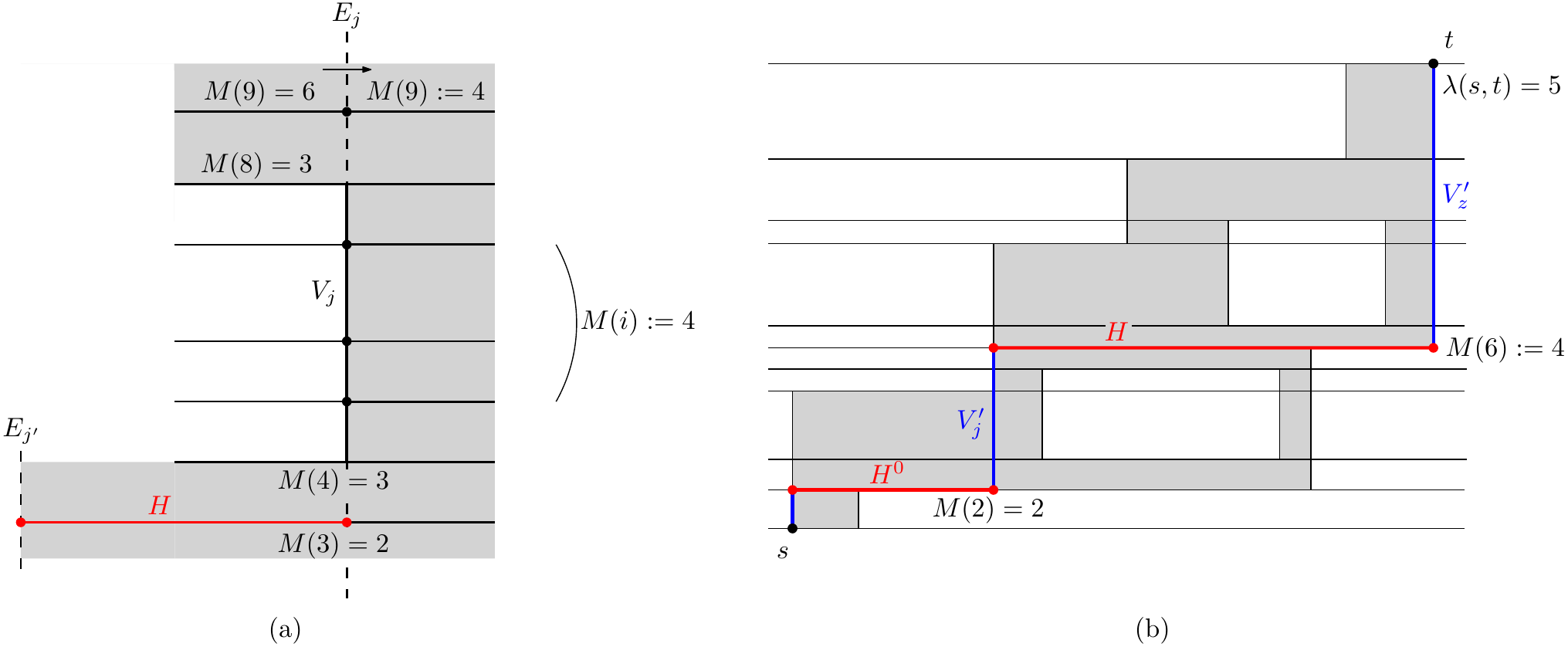}
    \caption{\small (a) A \merg event $E_j$. 
$M(3) = \min_{k\in[a',\alpha(j)]}M(k)$, which was updated from $E_{j'}$.
The baselines $H_i$ for $i=5,\ldots,7$ become active, and
$M(i)$ is updated to $4$ by $M(3)$ (Equation~\ref{eqn:merge}).
$M(9)$ is also updated by $M(3)$ (Equation~\ref{eqn:merge}). $H$ is the canonical segment for $E_j$.
(b) A minimum-link shortest path from $s$ to $t$
consisting of three vertical segments (blue)
and two horizontal segments (red). The red horizontal segments are canonical segments in $\xydom(s,t)$.
$\ml(s,t)$ is computed by $M(6)$ in the \termi event. $M(6)$ was updated in $E_j$
so there is a canonical segment $H$ for $E_z$. $H$ and $V'_z$ form a subpath of $\pi$.
    }
    \label{fig:pbase}
  \end{center}
\end{figure}

For instance, consider a \merg event $E_j$.
Recall that $\rangebbst$ has two disjoint ranges 
$[a',\alpha(j)]$ and $[\beta(j),b']$ with $a' < \alpha(j)$ and $\beta(j) < b'$. 
We update $M(i)$ using Equation~\ref{eqn:merge}.
Let $k^*\in[a',\alpha(j)]$ be the smallest index such that $M(k^*)$ equals $\min_{k\in[a',\alpha(j)]}M(k)$.
Assume that $M(k^*)$ was updated lately to the current value at an event $E_{j'}$
before $L$ encounters $E_j$.
Obviously, $x(V_{j'})< x(V_j)$.
We add a horizontal line segment $H$, which we call a canonical segment for $E_j$
with $x_1(H) = x(V_{j'})$, $x_2(H) = x(V_j)$ and $y(H) = y(H_{k^*})$.
See Figure~\ref{fig:pbase}(a).

We add one canonical segment for the \merg event $E_j$.
Likewise, we add one canonical segment per event of other types, except for the \origi event.
Since the $x$-coordinates of the events are distinct by the general position assumption,
the right endpoints of the canonical segments we add are also distinct.
Once the plane sweep algorithm is done,
by following lemma, we can report a shortest path that has $\ml(s,t)$ links.

\begin{lemma}\label{lem:canonical}
There is a minimum-link shortest path from $s$ to $t$
whose horizontal line segments are all canonical segments.
\end{lemma}
\begin{proof}
At the \termi event $E_z$, 
we have $\ml(s,t):=\min\{M(m),M(k^*)+1\}$,
where $k^*$ is the index satisfying $M(k^*) = \min_{k\in[\alpha(z),\beta(z)-1]}M(k)$.
If $M(m) \leq M(k^*)$, there is a canonical segment $H$ incident to $t$,
so $H$ becomes the horizontal segment of $\pi$ that is incident to $t$.
Otherwise, there is a canonical segment $H$ incident to $V_z$
with $y(H)=y(H_{k^*})$, and thus $H$ and $V'_z$ form a subpath 
of $\pi$, where $V'_z$ is the portion of $V_z$ with $y_1(V'_z) = y(H_{k^*})$ and $y_2(V'_z) = y(t)$.
For both cases, we can find the left endpoint of $H$ such that $x_1(H)=x(V_j)$.
We know there exists a canonical segment $H'$ for the event $E_j$,
so we do the above process for $H'$ and the $V'_j$ to form a subpath of $\pi$ recursively,
where $V'_j$ is the vertical line segment connecting the left endpoint of $H$ and the right endpoint of $H'$.
See Figure~\ref{fig:pbase}(b).
At the \origi event $E_1$,
we have a canonical segment $H^0$ with $x(H^0) = x(s)$.
Then the vertical line segment connecting $s$ and the left endpoint of $H^0$ forms a subpath of $\pi$.
Gluing all subpaths formed from above recursive process,
we finally obatin $\pi$ whose horizontal line segments are all canonical segments.
\end{proof}

$\xydom(s,t)$ can be obtained by ray shooting queries, each taking $O(\log n)$ time, 
using the data structure of Giora and Kaplan~\cite{giora2009}
with $O(n \log n)$ preprocessing time. 
Let $h$ be the number of holes in $\xydom(s,t)$,
and $o$ be the complexity of the outer boundary of $\xydom(s,t)$.
We can construct $\xydom(s,t)$ in $O((o+h) \log n)$
time using $O(o+h)$ space.

At each of the $z$ events, we remove and insert some ranges.
Because the ranges in $\rangebbst$ are disjoint by the definition of 
$\rangebbst$, we can insert and remove a range in $O(\log m)$ time by using
a simple balanced binary search tree for $\rangebbst$.
We also set or update some $M(i)$'s at each event.
For each event, if we know $\min_{k\in[a_1,b_1]} M(k)$ for a range $[a_1,b_1]$,
we can update $M(i)$ for $i\in[a_2,b_2]$ (with $b_1$ < $a_2$)
in time linear to the number of consecutive baselines from $H_{a_2}$ to $H_{b_2}$,
and the number of $M(i)$'s is $O(m)$.
Therefore, it takes $O(m)$ time to handle an event.
We use $O(m)$ space to maintain $\rangebbst$ and $M(i)$'s.
In total, we use $O(o+h+m)$ space to compute $\ml(s,t)$.
We can report a minimum-link shortest path using $O(z)$ canonical segments.
Thus, our algorithm takes $O((n+o+h)\log n + mz) = O(n^2)$ time and $O(o+h+m+z) = O(n)$ space.

\subsection{Reducing the time complexity}\label{sec:reduce-time}
To reduce the time complexity of our algorithm
to $O(\log m)$ for handling each event while keeping the space complexity to $O(n)$ space,
we build another balanced binary search tree $\segtree$, a variant of a segment tree in~\cite{de2008}.
The idea is to use $\segtree$ together with $\rangebbst$ to maintain $O(\log m)$ nodes corresponding to $O(m)$ $M(i)$'s efficiently,
instead of updating $M(i)$'s for each event immediately.

Each node $w$ of $\segtree$ corresponds to a sequence of baselines consecutive in their indices,
say from $\alpha(w)$ to $\beta(w)$ with $\alpha(w)\leq \beta(w)$.
Let $\lchild(w)$ and $\rchild(w)$ be the left child and the right child of $w$, respectively.
A leaf node $w$ corresponds to one baseline $H_i$, hence $\alpha(w) = \beta(w) = i$.
A nonleaf node $w$ corresponds to a sequence of baselines corresponding 
to the leaf nodes in the subtree rooted at $w$, and thus 
$\alpha(w) = \alpha(\lchild(w))$ and $\beta(w) = \beta(\rchild(w))$.
We say a node $w$ of $\segtree$ is \emph{inactive} 
if all baselines with indices from $\alpha(w)$ to $\beta(w)$
are inactive. Node $w$ is \emph{active} otherwise.

We can represent any range of indices using $O(\log m)$ nodes of $\segtree$
whose ranges are disjoint. We use $\wset[a,b]$ to denote the set of nodes of $\segtree$ 
such that $[a,b] = \cup_{w\in\wset[a,b]} [\alpha(w),\beta(w)]$ and $[\alpha(w),\beta(w)] \cap [\alpha(w'),\beta(w')]=\emptyset$ 
for any two nodes $w,w' \in \wset[a,b]$.
We define another set $\vset[a,b]$ of nodes $v$ of $\segtree$ such that $[\alpha(v),\beta(v)]\cap[a,b]\neq\emptyset$
and $[\alpha(v),\beta(v)]\not\subseteq[a,b]$.
Observe that the number of nodes in $\vset[a,b]$ is also $O(\log m)$.
See Figure~\ref{fig:segtree} for an illustration of $\wset[a,b]$ and $\vset[a,b]$.

\begin{figure}[t]
  \begin{center}
    \includegraphics[width=0.6\textwidth]{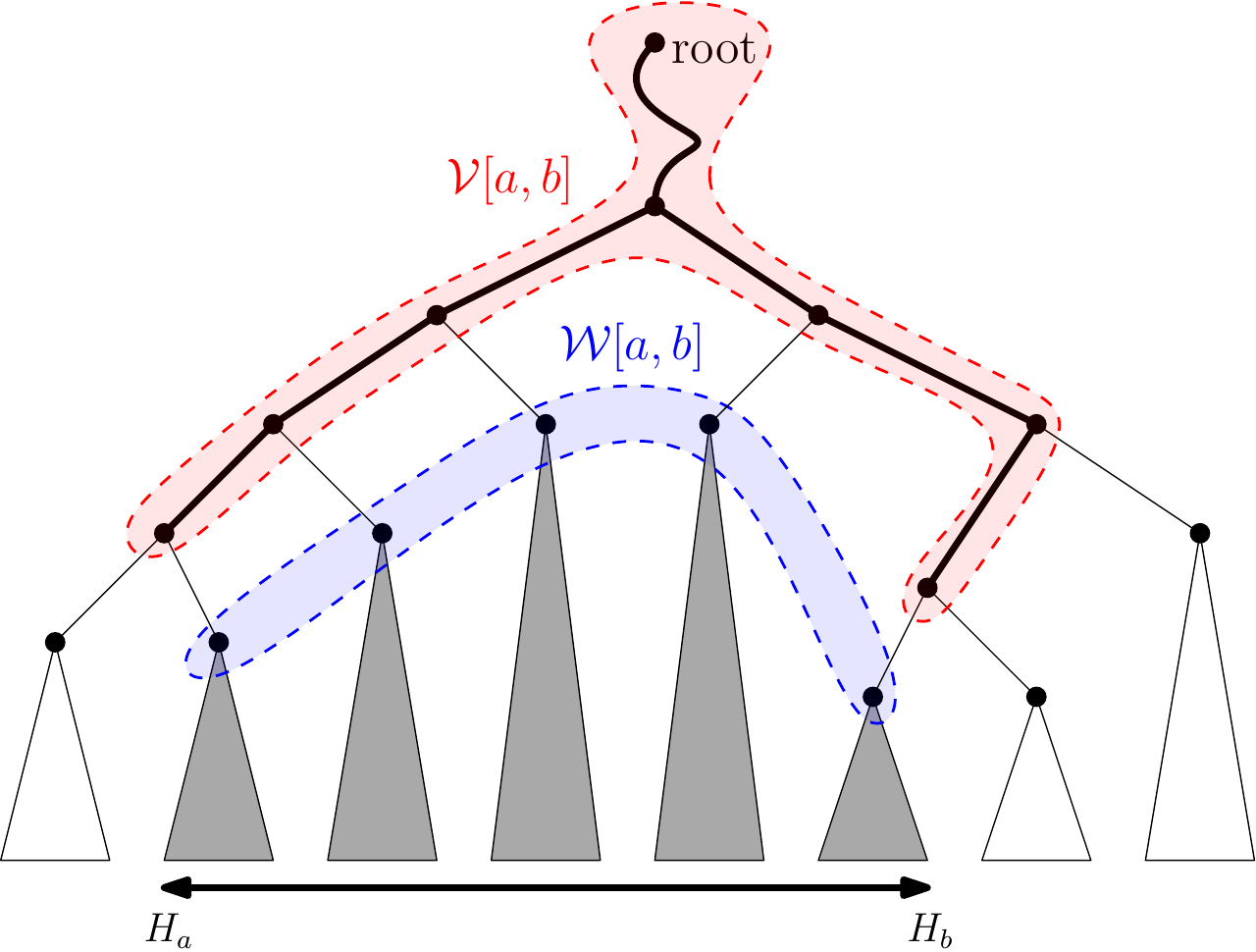}
    \caption{\small A balanced binary search tree $\segtree$. A range $[a,b]$ can be represented by
    the $O(\log m)$ nodes in $\wset[a,b]$. For each node $v\in\vset[a,b]$, $[\alpha(v),\beta(v)]\cap[a,b]\neq\emptyset$
and $[\alpha(v),\beta(v)]\not\subseteq[a,b]$.}
    \label{fig:segtree}
  \end{center}
\end{figure}

For a node $w\in\segtree$,
we define two values, $\ml(w)$ and $U(w)$ as
$\ml(w) = \min_{i\in[\alpha(w),\beta(w)]} M(i)$ and
$U(w) = \max_{i\in[\alpha(w),\beta(w)]} M(i)$.
We need $\ml(w)$ in computing a minimum-link shortest path,
while $U(w)$ is used for updating $\segtree$.
We initialize both $\ml(w)$ and $U(w)$ to $\infty$
for every node $w$ in $\segtree$.
We update these values stored at some nodes of $\segtree$
at an event during the plane sweep. 
At \origi and \termi events,
we update $\ml(w)$ and $U(w)$ for the leaf nodes of subtrees rooted at $w\in \wset[a,b]$ for a range $[a,b]$,
and also update for all $w$ in $\segtree$ using values of leaf nodes in bottom-up manner.
We process other types of events in the following way:
we find $\ml^* = \min_{w\in \wset[a_1,b_1]}\ml(w)$ for a range $[a_1,b_1]$,
and then update $\ml(u)$ and $U(u)$ for $u\in \wset[a_2,b_2]$ for another range $[a_2,b_2]$
disjoint from $[a_1,b_1]$. There are three cases.  
(1) If $u$ becomes inactive at the event, we set both $\ml(u)$ and $U(u)$ to $\infty$.
(2) If $u$ becomes active at the event, 
we set both $\ml(u)$ and $U(u)$ to $\ml^*+2$.
Observe that all baselines corresponding to $u$ become active at the event
since $u$ is in $\wset[a_2,b_2]$.
(3) If there is no status change in $u$, 
we set $U(u) := \min\{U(u), \ml^* +2\}$ and $\ml(u) := \min\{\ml(u),\ml^* +2\}$.
Once $\ml(u)$ and $U(u)$ are updated,
we also update $\ml(v)$ and $U(v)$ for $v\in \vset[a_2,b_2]$
in bottom-up manner. 

Observe that we update neither $\ml$ nor $U$ values
of the children of $u$ during the update of $\ml(u)$ and $U(u)$ for $u\in \wset[a_2,b_2]$.
Some nodes may have their $\ml$ and $U$ values outdated when they are used
for finding $\ml^*$ and updating $\ml$ an $U$ values of other nodes.
To resolve this problem, we update $\ml(v')$ and $U(v')$ 
for the children $v'$ of each node $v\in \vset[a,b]$
when we find $\wset[a,b]$ for every range $[a,b]$.
Note that to compute $\wset[a,b]$, we must compute $\vset[a,b]$.
By the definition of $\ml(v)$ and $U(v)$, we have
$\ml(v) = \min\{\ml(\lchild(v)),\ml(\rchild(v))\}$
and $U(v) = \max\{U(\lchild(v)),U(\rchild(v))\}$,
and $\ml(v) \leq U(v)$.

We update $\ml(v')$ and $U(v')$ for two children $v'\in\{\lchild(v),\rchild(v)\}$ of $v$.
If $\ml(v) = U(v)$,
by the definition of $\ml(v)$ and $U(v)$, we have $\ml(v') = \ml(v)$ and $U(v') = U(v)$ for all $v'$.
Therefore, we set $\ml(v'):=\ml(v)$ and $U(v'):=U(v)$ for all $v'$.

If $\ml(v) < U(v)$,
there are four subcases:
(1) $\ml(v) \neq \min_{v'}\ml(v')$
(2) $\ml(v)=\min_{v'}\ml(v')$ but $U(v) > \max_{v'}U(v')$,
(3) $\ml(v)=\min_{v'}\ml(v')$ but $U(v) < \max_{v'}U(v')$, and
(4) $\ml(v)=\min_{v'}\ml(v')$ and $U(v) = \max_{v'}U(v')$.
For the cases (1) and (2), we set $\ml(v'):=\ml(v)$ and $U(v'):=U(v)$ for all $v'$ because they are outdated.
For the case (3), for $v'$ satisfying $U(v) < U(v')$, we set $U(v'):=\min\{U(v'), U(v)\}$ and update $\ml(v')$ compared with $U(v')$.
For the case (4), we already use $\ml(v')$ and $U(v')$ to update $\ml(v)$ and $U(v)$ in bottom-up manner,
so they are not outdated and we do not change any values.

Recall that the number of nodes in $\wset[a,b]$ and $\vset[a,b]$ for a range $[a,b]$
is $O(\log m)$, and we can find them in $O(\log m)$ time 
since $\segtree$ is a balanced binary search tree with height $O(\log m)$.
See Chapter 10 in~\cite{de2008}.
For each node $w\in \wset[a,b]\cup \vset[a,b]$, only a constant number of nodes 
are affected by an update above,
and $\ml(u)$ or $U(u)$ for such node $u$
can be computed in constant time.
Thus, each query in $\segtree$ takes $O(\log m)$ time, and we can find one canonical segment
for each event in the same time because we have $\ml^*$.
By using this data structure, we can reduce the time complexity from $O(m)$ to $O(\log m)$ per event.
Since $\segtree$ uses $O(m)$ space, the total space complexity remains to be $O(n)$.
Thus, we have the following lemma.

\begin{lemma}\label{lem:xymonotone}
For a point $t$ in $\xydom$,
we can compute a minimum-link shortest path from $S$ to $t$
in $O(n\log n)$ time using $O(n)$ space.
\end{lemma}

\section{\texorpdfstring{$t$}{t} lies in \texorpdfstring{$\xdom$}{Dx} or \texorpdfstring{$\ydom$}{Dy}}
\label{sec:xmonotone}
In this section, we assume that $t$ lies in $\xdom$. Then every
shortest path from $S$ to $t$ is $x$-monotone by Lemma~\ref{lem:monotone}.
In case that $t$ lies in $\ydom$, every
shortest path is $y$-monotone and we can handle the case in a similar way.
Unlike the case of $t\in\xydom$, there can be a shortest path from $S$ to $t$
not contained in $\xdom$. 
See Figure~\ref{fig:xydomtree}(a).
However, we can compute a minimum-link shortest path from $S$ to $t$ 
using the algorithm in Section~\ref{sec:xymonotone} as a subprocedure.

\begin{figure}[t]
  \begin{center}
    \includegraphics[width=\textwidth]{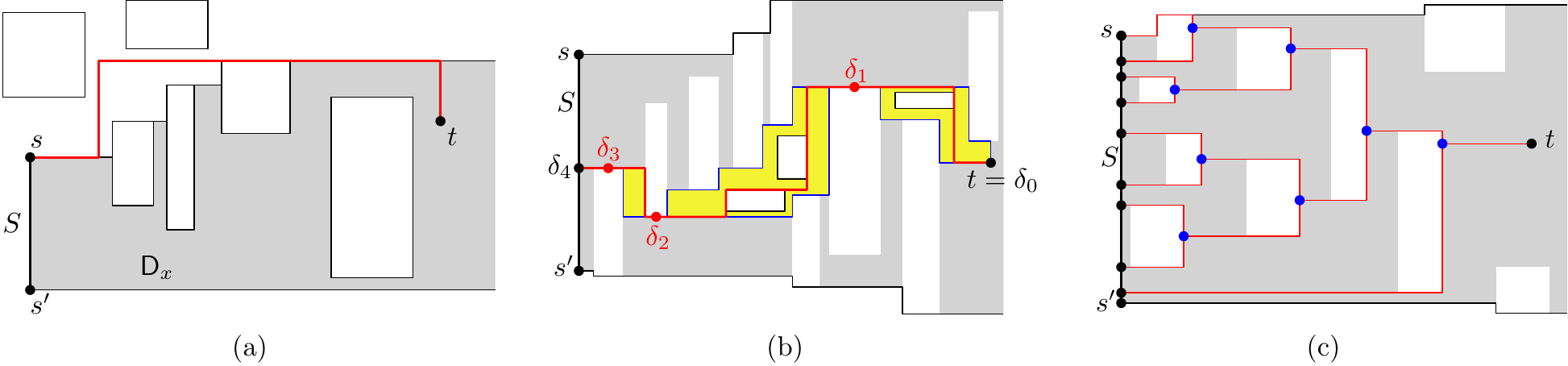}
    \caption{\small
(a) A minimum-link shortest path from $S$ to $t$.
It is not contained in $\xdom$.
(b) Every shortest path from $S$ to $t$ passes through
$\divp_i$ for $i=0,\ldots, 4$ and is contained in $\bigcup_{i=0,\ldots, 3}\xydom(\divp_i,\divp_{i+1})$.
(c) Shortest paths from $S$ to $t$ such that the closest pair of $S$ and $t$ is not unique.
    }
    \label{fig:xydomtree}
  \end{center}
\end{figure}

Let $\pi$ denote a minimum-link shortest path from $S$ to $t$ 
aligned to the baselines,
and let $s^*$ and $t$ be the two endpoints of $\pi$ with $s^*\in S$. 
By definition, $(s^*,t)$ is a closest pair of $S$ and $t$.
Since $\pi$ is $x$-monotone, it is a concatenation of $xy$-monotone paths 
such that every two consecutive $xy$-monotone subpaths of $\pi$ change their directions
between monotone increasing and monotone decreasing on
a horizontal segment, which we call a \emph{winder}, of $\pi$. 

\begin{lemma}\label{lem:winder}
Every winder of a shortest path from $S$ to
$t$ contains one entire horizontal side of a rectangle in $\rectset$
incident to $\xdom$.
\end{lemma}
\begin{proof}
Let $\pi$ be a shortest path from $S$ to $t$.
Assume that $\pi$ has a winder $H$,
which does not contain a horizontal side of a rectangle.
Since $H$ is a winder, the two consecutive
$xy$-monotone subpaths of $\pi$ sharing $H$ lie
in one side of the line containing $H$.
Without loss of generality, assume that both subpaths lie above the line
containing $H$. Then we can drag $H$ upward while shortening 
the vertical segments of $\pi$ incident to $H$, which results in a shorter 
path, a contradiction.

Now assume to the contrary that there is a winder $H$ of $\pi$ containing 
a horizontal side of a rectangle not incident to $\xdom$. By the general position,
$H$ does not contain a horizontal side of a rectangle incident to $\xdom$.
Then the subpath of $\pi\setminus\xdom$ containing $H$ is
not a shortest path connecting its two endpoints incident to $\xdom$. Thus, $\pi$ is not a shortest path
from $S$ to $t$, a contradiction.
\end{proof}

\subsection{Computing a minimum-link shortest path}
Consider the horizontal sides of rectangles contained in the winders of 
a minimum-link shortest path $\pi$. Let $g$ be the number of winders
of $\pi$, and let $\divp_i$
be the midpoint of the horizontal side contained in the $i$th winder 
in order along $\pi$ from $t$ to $s^*$.
We call such a midpoint a \emph{divider} of $\pi$.
For convenience, we let $t=\divp_0$ and $s^*=\divp_{g+1}$.
Then the subpath from $\divp_i$ to $\divp_{i+1}$ for $0\leq i\leq g$ of $\pi$
is $xy$-monotone by the definition of the winders of $\pi$.

By Lemma~\ref{lem:winder}, every winder of a minimum-link shortest path 
from $S$ to $t$ contains a horizontal side of a rectangle incident to $\xdom$.
Therefore, in the following 
we compute the dividers of a minimum-link shortest path
among the midpoints of the top and bottom sides 
of each rectangle incident to $\xdom$, and compute 
the $xy$-monotone paths connecting the dividers, in order,
which together form a minimum-link shortest path.

We compute $d(S,t)$ by a plane sweep algorithm and find
the dividers $\divp_1,\ldots,\divp_g$ of a minimum-link shortest path $\pi$ as follows.
For a rectangle $R\in\rectset$ incident to $\xdom$,
let $\med(R)$ and $\med'(R)$ denote
the midpoints of the top and bottom sides
of $R$, respectively. Then $\med(R)$ and $\med'(R)$
are candidates of the dividers of $\pi$. 
We consider each midpoint as an event during the sweep.

While sweeping $\xdom$ with a vertical line $L$ moving rightwards,
$L$ encounters $\med(R)$ and $\med'(R)$ of a rectangle $R$
at the same time.
Consider two horizontal rays, one from $\med(R)$ and one from $\med'(R)$, 
going leftwards.
We show how to handle the ray $\gamma$ from $\med(R)$.
The ray from $\med'(R)$ can be handled similarly.
Let $p_\gamma$ be the point of the vertical segment
on the boundary of $\xdom$ which $\gamma$ hits first.
If $p_\gamma\in S$, the shortest path from $S$ to $\med(R)$ 
is simply $p_\gamma\med(R)$.
If $p_\gamma$ lies in $\piru(s)$ (or $\pird(s')$),
every shortest path from $S$ to $\med(R)$ is $xy$-monotone.
For these two cases,  
we store at $\med(R)$ the distance $d(S,\med(R))$ and 
the closest point (one of $p_\gamma$, $s$, or $s'$) of $S$ from $\med(R)$.
Consider the case that $p_\gamma$ is in the right side of a rectangle $R'\in\rectset$ incident to 
$\xdom$.
We already have $d(S,\med(R'))$ stored at $\med(R')$ and $d(S,\med'(R'))$ stored at $\med'(R')$ during the plane sweep.
Observe that every shortest path from $\med(R)$ to $\med(R')$ or to $\med'(R')$ is 
$xy$-monotone, and the closest point $s^o$ of $S$ from $\med(R)$ is the closest point of $S$
from $\med(R')$ or from $\med'(R')$. Since
$d(s^o,\med(R))=\min\{d(S,\med(R'))+d(\med(R'), \med(R)), d(S,\med'(R'))+d(\med'(R'), \med(R))\}$
by definition, we can compute $s^o$ and $d(s^o,\med(R))$ in constant time.

When $L$ encounters $t$, we again consider a horizontal ray $\gamma$ from $t$ going leftwards
and the point $p_\gamma$ of the vertical segment on the boundary of $\xdom$ which $\gamma$ hits first.
If $p_\gamma\in S$, the minimum-link shortest path is simply $p_\gamma t$ and we are done.
For $p_\gamma$ lying on a side of a rectangle $R\in\rectset$ incident to $\xdom$,
if $d(t,\med(R)) + d(S,\med(R)) < d(t,\med'(R)) + d(S,\med'(R))$ (or the other way around without equality),
we conclude there is no shortest path from $S$ to $t$ passing through $\med'(R)$ (or through $\med(R)$).
Assume that 
$\pi$ passes through $\med(R)$. By the general position assumption, $y(\med(R)) > y(t)$.
Let $R'$ be the rectangle incident to $\xdom$ 
that the horizontal ray from $\med(R)$ going leftwards hits first.
If $d(S,\med(R')) + d(\med(R'),\med(R)) > d(S,\med'(R')) + d(\med'(R'), \med(R))$
or there is no such rectangle $R'$, $\med(R)$ is a divider of $\pi$.
Moreover, $\med(R)$ is the first divider $\divp_1$ of $\pi$ from $t$, 
and thus every shortest path from $\divp_1=\med(R)$ to $t$ is $xy$-monotone.
Therefore, we construct $\xydom(\med(R),t)$ and apply the algorithm in Section~\ref{sec:xymonotone}.
Then we apply this procedure from $\med(R)$, recursively,
and compute every $xy$-monotone subpath of $\pi$
using canonical segments by Lemma~\ref{lem:canonical}, 
and glue them into one to form $\pi$.
See Figure~\ref{fig:xydomtree}(b).
Finally we obtain $g+1$ $xy$-monotone paths
with dividers $\divp_0=t,\divp_1,\ldots,\divp_{g+1}=s^*\in S$.

During the plane sweep, we find in $O(\log n)$ time the first rectangle hit by the horizontal ray $\gamma$ 
emanating from a midpoint of a rectangle going leftwards
using the data structure supporting ray shooting queries by Giora and Kaplan~\cite{giora2009}.
Thus, it takes $O(n \log n)$ time for ray shootings from midpoints in total.
It takes $O(K_i)$ time to find a divider $\divp_i$,
where $K_i$ is the number of the recursion depth of the algorithm to compute 
$\divp_i$ from $\divp_{i-1}$.
As shown in Section~\ref{sec:xymonotone},
computing an $xy$-monotone minimum-link shortest path from $\divp_i$ to $\divp_{i-1}$ 
takes $O(D_i \log D_i)$ time with $O(D_i)$ space after $O(n \log n)$-time preprocessing,
where $D_i$ is 
the number of the baselines defined by the rectangles
incident to $\xydom(\divp_i,\divp_{i-1})$.
Observe that 
$\Sigma_{1\leq i\leq g+1} K_i = O(n)$, and $\Sigma_{1\leq i\leq g+1} D_i = O(n)$
because 
the regions $\xydom(\divp_i,\divp_{i-1})$'s are disjoint in their interiors.
Thus, the total time complexity is $O(n \log n)$ and the total space complexity is $O(n)$.

\begin{lemma}\label{lem:xmonotone}
For a point $t$ in $\xdom\cup\ydom$,
we can compute a minimum-link shortest path from $S$ to $t$
in $O(n\log n)$ time using $O(n)$ space.
\end{lemma}

\subsection{Combinatorially distinct shortest paths}\label{sec:com-distinct}
There can be two shortest paths from $S$ to $t$, 
one passing through $\med(R)$ and one passing
through $\med'(R)$ for a rectangle $R$.
In this case, we have $d(t,\med(R)) + d(S,\med(R)) = d(t,\med'(R)) + d(S,\med'(R))$,
which can be found in handling the midpoints of $R$ during the plane sweep.
Observe that this equality may occur multiple times in finding dividers of 
a minimum-link shortest path. Thus we need to devise an efficient way of
maintaining all sequences of dividers, each of which may define 
a shortest path.
See Figure~\ref{fig:xydomtree}(c).
In this section, we show how to maintain these sequences of dividers and
how to find a minimum-link shortest path 
without increasing the time and space complexities in Lemma~\ref{lem:xmonotone}.

We say two shortest paths, $\pi_1$ and $\pi_2$, from $S$ to $t$ are \emph{combinatorially distinct} 
if the sequence of dividers for $\pi_1$ and the sequence of dividers for $\pi_2$ are different.
Let $\Pi$ be the set of all combinatorially distinct sequences of dividers from $t$ to $S$ for shortest paths,
since we let $t = \divp_0$ for convenience.
Assume that a divider $\divp$ appearing in a sequence of $\Pi$
is the midpoint of the bottom side of a rectangle.
Observe that $\pilu(\divp)$ passes through dividers
including $\divp$ that are consecutive in a sequence of $\Pi$.
We denote by $f(\divp)$ one with smallest $x$-coordinate among these dividers.
Observe that $f(\divp)$ is uniquely defined for $\divp$ with $x(f(\divp)) < x(\divp)$, and 
it is the midpoint of the top side of another rectangle.
We construct $\xydom(f(\divp),\divp)$ to compute the subpath of a minimum-link shortest path 
from $f(\divp)$ to $\divp$.

\begin{figure}[t]
  \begin{center}
    \includegraphics[width=\textwidth]{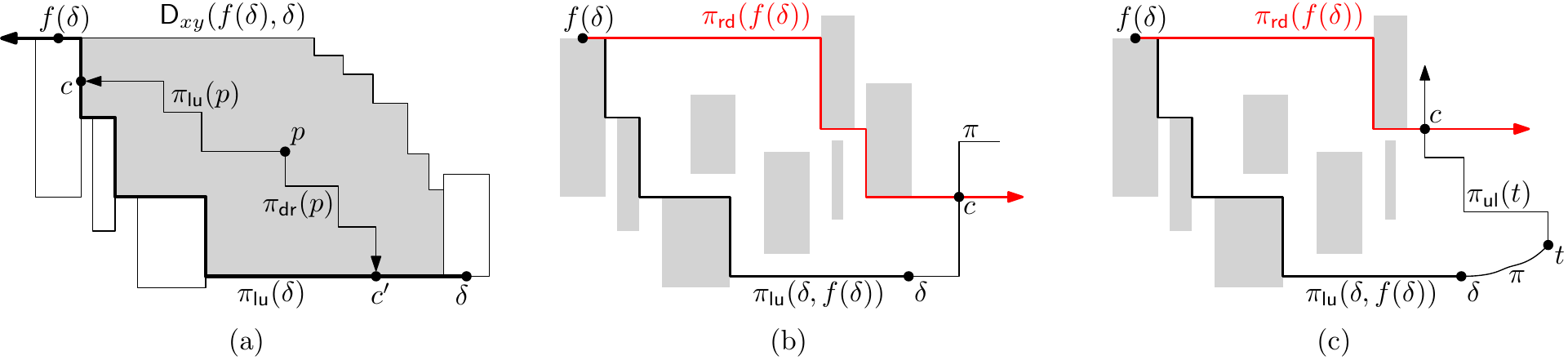}
    \caption{\small 
Proof of Lemma~\ref{lem:xyshortest} and~\ref{lem:outerchain}.
    }
    \label{fig:lemma89}
  \end{center}
\end{figure}

\begin{lemma}\label{lem:xyshortest}
For any point $p$ in $\xydom(f(\divp),\divp)$, there are an $xy$-monotone path
from $f(\divp)$ to $p$ and an $xy$-monotone path from $p$ to $\divp$, which
are shortest among paths connecting the points.
\end{lemma}
\begin{proof}
For any point $p$ in $\xydom(f(\divp),\divp)$, let $c$ be a point in the intersection 
$\pilu(\divp)\cap\pilu(p)$. Then the path obtained by concatenating the subpath
of $\pilu(\divp)$ from $f(\divp)$ to $c$ and the subpath of $\pilu(p)$ from $c$ to $p$
is $xy$-monotone, and it is shortest among all paths from $f(\divp)$ to $p$.
Similarly, let $c'$ be a point in the intersection 
$\pilu(\divp)\cap\pidr(p)$. Then the path obtained by concatenating the subpath
of $\pidr(p)$ from $p$ to $c$ and the subpath of $\pilu(\divp)$ from $c'$ to $\divp$
is $xy$-monotone, and it is shortest among all paths from $p$ to $\divp$.
See Figure~\ref{fig:lemma89}(a).
\end{proof}

\begin{lemma}\label{lem:outerchain}
$\xybd(\divp,f(\divp))$ is $\pilu(\divp,f(\divp))$ and 
$\xybd(f(\divp),\divp)$ is $\pird(f(\divp),\divp)$.
\end{lemma}
\begin{proof}
By the definitions of $f(\divp)$ and $\xydom(f(\divp),\divp)$, 
$\xybd(\divp,f(\divp))$ is $\pilu(\divp,f(\divp))$.

Let $\pi$ be a shortest path from $S$ to $t$ that
contains $\pilu(\divp,f(\divp))$ as a subpath.
Assume that $\pird(f(\divp))$ does not pass through $\divp$.
If $\pird(f(\divp))\cap\pi\neq\emptyset$ except for $f(\divp)$,
let $c$ be the last point of $\pird(f(\divp))\cap\pi$ along $\pird(f(\divp))$
from $f(\divp)$. Since $\pird(f(\divp))$ does not pass through $\divp$,
$\pird(f(\divp))$ does not intersect the portion of $\pi$ from $f(\divp)$ to $\divp$. 
Thus, $c$ is on the portion of $\pi$ from $\divp$ to $t$. 
See Figure~\ref{fig:lemma89}(b).
If $\pird(f(\divp))\cap\pi=\emptyset$ except for $f(\divp)$,
let $c$ be the last point of $\pird(f(\divp))\cap\piul(t)$ along $\pird(f(\divp))$
from $f(\divp)$.
See Figure~\ref{fig:lemma89}(c).
Since $\pird(f(\divp),c)$ is $xy$-monotone, 
it is shorter than the portion of $\pi$ from $f(\divp)$ to $c$ in both cases.  
Thus, we can get a path from $S$ to $t$ shorter than $\pi$ by
replacing the portion of $\pi$ from $f(\divp)$ to $c$ with $\pird(f(\divp),c)$, a contradiction.
In other words, $\pird(f(\divp))$ pass through $\divp$,
so it implies that $\xybd(f(\divp),\divp)$ is $\pird(f(\divp),\divp)$.
\end{proof}

\begin{lemma}\label{lem:not_interior}
There is no divider on the inner boundary of $\xydom(f(\divp),\divp)$.
\end{lemma}
\begin{proof}
Assume to the contrary that there is a divider $\divp'$ 
on the inner boundary of $\xydom(f(\divp),\divp)$. 
Since there are two $xy$-monotone paths, one from $f(\divp)$ to $\divp'$
and one from $\divp'$ to $\divp$ by Lemma~\ref{lem:xyshortest}, 
there is an $xy$-monotone path $\pi'$ from $f(\divp)$ to $\divp$
that passes through $\divp'$. 
Thus, there is a shortest path $\pi$ from $S$ to $t$ that contains $\pi'$ as a subpath.
 
Since there is a sequence of dividers containing $\divp'$ in $\Pi$, there is a shortest path $\pi^*$
from $S$ to $t$ that uses $\divp'$ as a divider, that is, two $xy$-monotone subpaths of $\pi^*$ 
change their directions at $\divp'$. 
Let $\pi''$ be the subpath of $\pi^*\cap\xydom(f(\divp),\divp)$ that
passes through $\divp'$, and 
let $c$ and $c'$ be the endpoints of $\pi''$ with $x(c)\leq x(c')$.
If both $c$ and $c'$ are on $\pilu(\divp)$,
then by replacing $\pi''$ of $\pi^*$ by the portion of $\pilu(\divp)$ between $c$
and $c'$ we can get a path from $S$ to $t$ shorter than $\pi^*$, a contradiction.
Similarly, for the case that both $c$ and $c'$ are on $\pird(f(\divp))$
we can get a shorter path  both by replacing $\pi''$ of $\pi^*$ 
by the portion of $\pird(f(\divp))$ between $c$ and $c'$.

Consider the case that $c$ is on $\pilu(\divp)$ and $c'$ is on $\pird(f(\divp))$.
Let $\pi^+$ be the concatenation of the subpath of $\pi$ from $S$ to $\divp'$
and the subpath of $\pi^*$ from $\divp'$ to $t$,
and $\pi^-$ be the concatenation of the subpath of $\pi^*$ from $S$ to $\divp'$
and the subpath of $\pi$ from $\divp'$ to $t$.
Observe that $\pi^+$ and $\pi^-$ should be also shortest paths from $S$ to $t$.
If $\divp'$ is the midpoint of the bottom side of a rectangle,
by replacing the subpath from $f(\divp)$ to $c'$ of $\pi^+$ with
an $xy$-monotone path
from $f(\divp)$ to $c'$ along $\pird(f(\divp))$, we can get a
path from $S$ to $t$ shorter than $\pi^+$. See Figure~\ref{fig:innerbound}(a).
If $\divp'$ is the midpoint of the top side of a rectangle,
by replacing the subpath from $c$ to $\divp$ of $\pi^-$ with
an $xy$-monotone path
from $c$ to $\divp$ along $\pilu(\divp)$, we can get a
path from $S$ to $t$ shorter than $\pi^-$. See Figure~\ref{fig:innerbound}(b).

Consider the case that $c$ is on $\pird(f(\divp))$ and $c'$ is on $\pilu(\divp)$.
Since $\pi''$ also uses $\divp'$ as a divider,
observe that the subpath of $\pi^*$ from $c$ to $c'$ is not $xy$-monotone.
By replacing the subpath from $f(\divp)$ to $c'$ of $\pi^+$ with
an $xy$-monotone path
from $f(\divp)$ to $c'$ along $\pilu(\divp)$, we can get a
path from $S$ to $t$ and we let $d_1$ be the length of the path.
By replacing the subpath from $\divp$ to $c$ of $\pi^-$ with
an $xy$-monotone path
from $\divp$ to $c$ along $\pird(f(\divp))$, we can get a
path from $S$ to $t$ and we let $d_2$ be the length of the path.
Then $\min\{d_1,d_2\} < d(S,t)$, so we can get a path shorter than
either $\pi^+$ or $\pi^-$. See Figure~\ref{fig:innerbound}(c).
\end{proof}

\begin{figure}[t]
  \begin{center}
    \includegraphics[width=\textwidth]{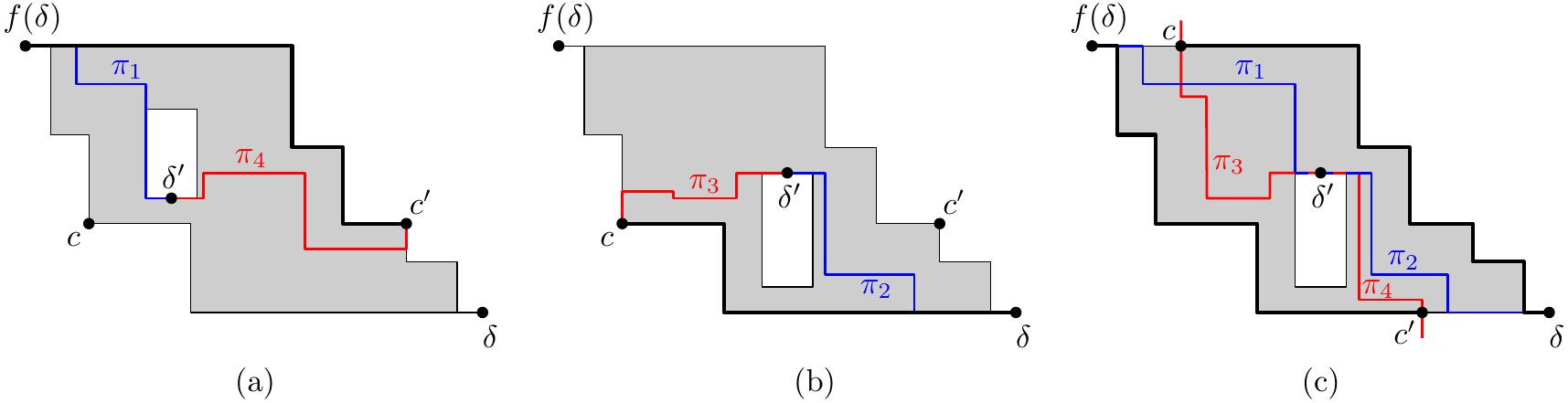}
    \caption{\small 
Proof of Lemma~\ref{lem:not_interior}. 
Blue paths $\pi_1$ and $\pi_2$ are subpaths of $\pi$, and red paths $\pi_3$ and $\pi_4$ are subpaths of $\pi^*$.
(a) If $\divp'$ is the midpoint of the bottom side of a rectangle,
an $xy$-monotone path from $f(\divp)$ to $c'$ along $\pird(f(\divp))$ (thick path)
is shorter than the concatenation of $\pi_1$ and $\pi_4$ (subpath of $\pi^+$).
(b) If $\divp'$ is the midpoint of the top side of a rectangle,
an $xy$-monotone path from $c$ to $\divp$ along $\pilu(\divp)$ (thick path)
is shorter than the concatenation of $\pi_3$ and $\pi_2$ (subpath of $\pi^-$).
(c) One of thick paths is shorter than one of two concatenations:
one is of $\pi_1$ and $\pi_4$ (subpath of $\pi^+$), the other is of $\pi_3$ and $\pi_2$ (subpath of $\pi^-$).
    }
    \label{fig:innerbound}
  \end{center}
\end{figure}

\begin{lemma}\label{lem:same-farthest}
If there is a divider $\divp'$ lying on $\xybd(f(\divp),\divp)$,
$f(\divp')$ and $f(\divp)$ are the same.
\end{lemma}
\begin{proof}
Let $\divp'$ be a divider lying on $\xybd(f(\divp),\divp)$.
Observe that $\divp'$ is the midpoint of the bottom side of a rectangle by Lemma~\ref{lem:outerchain}.
By definition, there is a shortest path $\pi$ 
from $S$ to $t$ passing through $f(\divp')$ and $\divp'$.
Since $\divp'$ lies on $\xybd(f(\divp),\divp)$,
$\pilu(\divp')$ passes through $f(\divp)$ and $f(\divp')$ by Lemma~\ref{lem:outerchain}.

Assume to the contrary that $f(\divp') \neq f(\divp)$.
By definition, $f(\divp')$ does not lie on $\xybd(\divp,f(\divp))$.
By Lemma~\ref{lem:not_interior}, $f(\divp')$ does not lie on the inner boundary 
of $\xydom(f(\divp),\divp)$.
Thus $f(\divp')$ is not incident to $\xydom(f(\divp),\divp)$.
Since $\pilu(\divp')$ passes through both $f(\divp)$ and $f(\divp')$,
$d(f(\divp'),\divp') = d(f(\divp'),f(\divp))+d(f(\divp),\divp')$.
We observe that 
$d(S,f(\divp)) < d(S,f(\divp'))+d(f(\divp'),f(\divp))$
since otherwise $f(\divp')$ and $\divp$ are consecutive in a sequence of $\Pi$,
so it violates the definition of $f(\divp)$.

Adding $d(f(\divp),\divp')$ to the both sides of the inequality,
we have $d(S,\divp') < d(S,f(\divp'))+d(f(\divp'),\divp')$.
This contradicts that $\pi$ is a shortest path from $S$ to $t$.
\end{proof}

For a fixed divider $\divp$, there are dividers each of which appear before $\divp$ consecutively in a sequence in $\Pi$.
Among those dividers, we let $g(\divp)$ be the divider with the largest $x$-coordinate.
Let $\divp$ be a divider satisfying $g(f(\divp)) = \divp$.
By Lemma~\ref{lem:outerchain}, $\xydom(f(\divp),\divp)$ is bounded by 
$\pilu(\divp,f(\divp))$ and $\pird(f(\divp),\divp)$.
By Lemmas~\ref{lem:outerchain} and~\ref{lem:same-farthest},
$\xydom(f(\divp),\divp)$ contains all $\xydom(\divp_i,\divp_j)$'s,
where $\divp_i$ is a divider lying on $\pilu(\divp,f(\divp))$
and $\divp_j$ is a divider lying on $\pird(f(\divp),\divp)$
such that $\divp_i$ and $\divp_j$ are consecutive in a sequence of $\Pi$.
After constructing $\xydom(f(\divp),\divp)$, we can compute a minimum-link shortest path
among all shortest paths from $\divp_i$ to $\divp_j$ using the plane sweep algorithm in Section~\ref{sec:xymonotone}. 

We first find a divider $\divp$ such that $g(f(\divp)) = \divp$
appearing in a sequence of $\Pi$.
By choosing a divider in decreasing order of the $x$-coordinate,
we can easily find such $\divp$.
We find $f(\divp)$ and construct $\xydom(f(\divp),\divp)$.
In $\xydom(f(\divp),\divp)$, we apply the algorithm in Section~\ref{sec:xymonotone}.
In the algorithm, 
each divider, including $\divp$, lying on $\xybd(f(\divp),\divp)$ is considered as an \origi event,
and each divider, including $f(\divp)$, lying on $\xybd(\divp,f(\divp))$ is considered as an \termi event.
See Figure~\ref{fig:xyregions4.1}(a).

To apply the plane sweep algorithm in Section~\ref{sec:xymonotone},
we must have $\ml(\divp',t)$ in advance
for each divider $\divp'$ considered as an \origi event.
When the sweep line encounters $\divp'$,
we update $M(a):=\min(M(a),\ml(\divp',t))$,
where $a$ is the index of the horizontal baseline incident to $\divp'$.
When the sweep line encounters $\divp'$ considered as a \termi event,
we can compute $\ml(\divp',t) := M(a)$,
where $a$ is an index of the horizontal baseline incident to $\divp'$.
Then $\ml(\divp',t)$ can be used when $\divp$ is considered as an \origi event
in other $xy$-monotone subregions.

Recall that $t$ and the closest points of $S$ from $t$ are not divider,
but they also construct $xy$-monotone subregions.
If $t=\divp$, Lemma~\ref{lem:outerchain} does not hold, but 
$\xydom(f(t),t)$ has no divider on $\xybd(f(t),t)$ except $f(t)$ and $t$.
Therefore, we do not have to change the \origi event in $\xydom(f(t),t)$.
If $f(\divp)$ is the closest point $s$ of $S$ from $t$,
Lemma~\ref{lem:outerchain} does not hold, but $\xydom(s,g(s))$ has no divider on
$\xybd(g(s),s)$ except $g(s)$ and $s$.
Therefore, we do not have to change the \termi event in $\xydom(s,g(s))$.

Therefore, to compute $\ml(\divp',t)$ at the \termi event, where $\divp'$ lies on $\xybd(\divp,f(\divp))$,
we have to know $\ml(\divp'',t)$ at the \origi event, where $\divp''$ lies on $\xybd(f(\divp),\divp)$.
It implies that there is an order among $xy$-monotone subregions to compute a minimum-link shortest path correctly.
With the order, we can construct a directed acyclic graph, which is a dual graph
of the $xy$-monotone subregions.
Each node $v$ of the graph corresponds to $\xydom(f(\divp),\divp)$.
We connect a directed edge from $u$ to $v$ if the two subregions corresponding to $u$ and $v$ 
are adjacent, and the subregion corresponding to $u$ has a divider $\divp$ as a \termi event,
and the subregion corresponding to $v$ has $\divp$ as an \origi event.
See Figure~\ref{fig:xyregions4.1}(b).

\begin{figure}[t]
  \begin{center}
    \includegraphics[width=\textwidth]{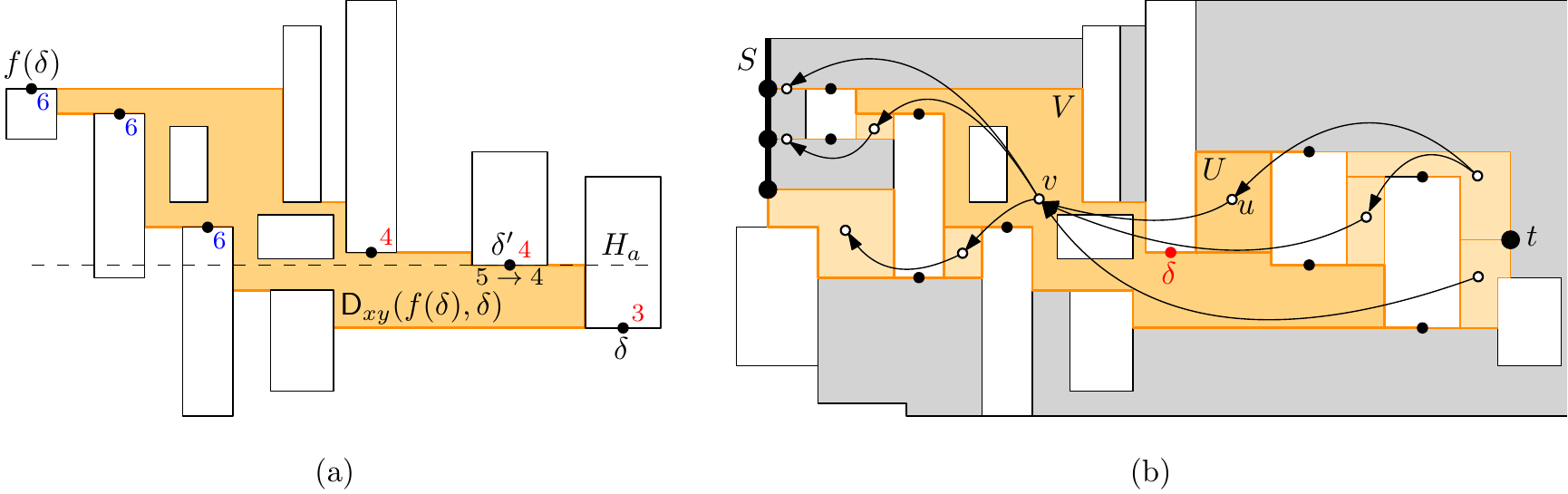}
    \caption{\small
(a) $\xydom(f(\divp),\divp)$ where $\divp$ is the divider such that $g(f(\divp)) = \divp$.
Points including $\divp$ and $\divp'$ lying on $\xybd(f(\divp),\divp)$ are dividers as \origi events,
and points including $f(\divp)$ lying on $\xybd(\divp,f(\divp))$ are dividers as \termi events.
During the plane sweep, $M(a) = 5$ before the sweep line encounters $\divp'$.
When the line encounters $\divp'$, $M(a)$ is updated since we know $\ml(\divp,t)$ is 4.
(b) A directed acyclic graph with constructed $xy$-monotone subregions.
Each $xy$-monotone subregions correspond to vertices of the graph represented as circles.
Both the outer boundaries of the two regions $U$ and $V$ corresponded to $u$ and $v$
contains $\divp$.
$\divp$ is an \termi event in $U$
and an \origi event in $V$.
To apply the plane sweep algorithm in $V$,
$\ml(\divp,t)$ should be computed first,
so there exists a directed edge $(u,v)$.
    }
    \label{fig:xyregions4.1}
  \end{center}
\end{figure}



Then we can compute $\ml(S,t)$ using a sequence of $xy$-monotone subregions
corresponding to a path in the dual graph.
Recall that during the plane sweep for each $xy$-monotone subregion,
we construct canonical segments to find a minimum-link shortest path
whose horizontal line segments are all canonical segments.
Since the $xy$-monotone subregions are disjoint in their interiors, 
we can report an $xy$-monotone path
using canonical segments by Lemma~\ref{lem:canonical},
and glue them to get a minimum-link shortest path from $S$ to $t$.

Since one rectangle has at most two dividers,
there are $O(n)$ dividers and $O(n)$ closest pairs of $S$ and $t$.
By Lemmas~\ref{lem:outerchain} and~\ref{lem:same-farthest} with the property of $\xydom(f(\divp),\divp)$,
there are at most two $xy$-monotone subregions incident to 
a divider that we construct during the plane sweep. 
Thus, there are four such subregions incident to a rectangle. 
Also, those subregions are disjoint in their interiors by Lemma~\ref{lem:not_interior}.

\begin{lemma}\label{lem:complexity}
During the plane sweep,
we construct $O(n)$ $xy$-monotone subregions defined by pairs of dividers
whose total complexity is $O(n)$. By using these subregions, we can compute
a minimum-link shortest path in $O(n\log n)$ time using $O(n)$ space.
\end{lemma}

\section{Extending to a line segment \texorpdfstring{$T$}{T}}
\label{se:extension}
Consider the case that the target is not just a point but an axis-aligned line segment $T$.
We explain how the algorithm presented in previous sections works for $T$.
Assume that $T$ is a vertical line segment and $x(S) < x(T)$.
We partition the domain $\dom$ into eight regions using the eight monotone paths $\pi_\alpha$'s from $S$ 
defined in Section~\ref{sec:eightpath}.
Then $T$ intersects at most five regions $\xdom^1$, $\xydom^1$, $\xydom^4$, $\ydom^1$, 
and $\ydom^2$. For the portion $T'$ of $T$ contained in each region, 
we compute a minimum-link shortest path from $S$ to $T'$.

For the portion of $T$ contained in a region of $\xydom^1, \xydom^4, \ydom^1$ and $\ydom^2$,
the closest point of $S$ from $T'$ is an endpoint of $S$
and the closest point in $T'$ from $S$ is an endpoint of $T'$ by Lemma~\ref{lem:monotone}.
Thus we just apply the algorithms in Sections~\ref{sec:xymonotone} and~\ref{sec:xmonotone} 
for the corresponding endpoints of $S$ and $T'$.

Consider the case that $T'\subset\xdom^1$.
A minimum-link shortest path from $S$ to $T'$ connects $S$ and an endpoint of $T'$ or
the intersection point $t'$ of $T'$ with a horizontal baseline of $\xdom$.
We can compute the distance from $S$ to two endpoints of $T'$ using the algorithm in Section~\ref{sec:xmonotone}.
There are $O(n)$ intersection points on $T'$ with horizontal baselines of $\xdom$.
During the plane sweep, we have $d(S, \med(R))$ and $d(S, \med'(R))$
for each hole $R$ of $\xdom$ such that the horizontal baselines defined by $R$ intersects $T'$.
Thus, we can compute the distance from $S$ to each intersection point $t'$ on $T'$ 
after the plane sweep. Then we obtain all the closest pairs of $S$ and $T'$.

If there is only one closest pair,
or the closest point of $T'$ from $S$ is the same for all closest pairs,
we can compute a minimum-link shortest path from $S$
to $T'$ as we do in Section~\ref{sec:xmonotone}.
Otherwise, let $t_1$ and $t_2$ be the closest points of $T'$ from $S$.
Since the two horizontal rays from $t_1$ and $t_2$ going leftwards
pass through two different dividers, Lemmas~\ref{lem:not_interior} and~\ref{lem:same-farthest} 
hold.
We compute all dividers of shortest paths from $S$ to the closest points of $T'$,
and compute a minimum-link shortest path from $S$ to $T'$ in a way similar to 
the one in Section~\ref{sec:com-distinct}. See Lemma~\ref{lem:complexity}.

We can compute the portions $T'$ of $T$ contained in each of the five regions in $O(\log n)$ time
using binary search along each path $\pi_\alpha$ and computing an intersection of $T$ and $\pi_\alpha$.
For $T'\subset\xdom^1$, we can find the closest pairs in $O(n \log n)$ time
if we use the ray shooting structure of Giora and Kaplan~\cite{giora2009}.
For each $T'$ we use our algorithm in Sections~\ref{sec:xymonotone} and~\ref{sec:xmonotone} 
with $O(n \log n)$ time and $O(n)$ space,
and eventually find a minimum-link shortest path from $S$ to $T$ 
by choosing $\min \ml(S,T')$ for all $T'$.

\begin{lemma}\label{lem:resultrectangle}
Given two axis-aligned line segments $S$ and $T$ in 
a rectangular domain with $n$ disjoint rectangular obstacles in the plane,
we can compute a minimum-link shortest path from $S$ to $T$ in $O(n\log n)$ time
using $O(n)$ space.
\end{lemma}

\section{Extending to box-disjoint rectilinear polygons}
\label{se:rect.polygons}
We show how to extend our algorithm in previous sections so that it handles 
box-disjoint rectilinear polygons.
Let $\rectset_P$ be a set of box-disjoint rectilinear polygons, and let $\bbox(P)$
denote the bounding box of a polygon $P\in\rectset_P$.
We use $\rdom := \realp - \cup_{P\in\rectset_P}P$ to denote a \emph{box-disjoint rectilinear domain}
induced by $\rectset_P$ in the plane.
A set $Q$ is \emph{rectilinear convex} if and only if any line parallel to the $x$- or $y$-axis intersects $Q$
in at most one connected component.
The rectilinear convex hull of $P$, denoted by $\ch(P)$, is the common intersection of
all rectilinear convex sets containing $P$. 

\subsection{Both \texorpdfstring{$S$}{S} and \texorpdfstring{$T$}{T} disjoint from the bounding boxes}\label{sec:bbox-disjoint}
Consider the case that both $S$ and $T$ are disjoint from the rectangles $\bbox(P)$ for $P\in\rectset_P$. 
Then no shortest path 
intersects the interior of $\ch(P)$ for $P\in\rectset_P$. If there is a shortest path $\pi$ intersecting the 
interior of $\ch(P)$ for a rectilinear polygon $P\in\rectset_P$, $\pi$ can be shortened by replacing
each connected portion of $\pi$ contained in the interior with the boundary curve of
$\ch(P)$ between the endpoints of the portion, a contradiction.
Thus, we replace each polygon $P$
with $\ch(P)$ and find a minimum-link shortest path from $S$ to $T$ avoiding $\ch(P)$'s.
We assume that each polygon $P\in\rectset_P$ is rectilinear convex in this subsection.
If there is a shortest path $\pi$ from $S$ to $T$ intersecting $\bbox(P)$ for $P\in\rectset_P$,
the subpath $\pi\cap\bbox(P)$ can be replaced
with a subpath along the boundary of $\bbox(P)$ without increasing the length.
This implies that there is a shortest path from $S$ to $T$ avoiding $\bbox(P)$ for all $P\in\rectset_P$.
From Lemma~\ref{lem:monotone}, every shortest path from $S$ to $T$ avoiding $\bbox(P)$ for all $P\in\rectset_P$
is either $x$-, $y$-, or $xy$-monotone.
The two subpaths have same length and endpoints,
so they have the same monotonicity: One is $X$-monotone if and only if the other is $X$-monotone, for $X\in\{x,y,xy\}$.
Therefore, every shortest path from $S$ to $T$ contained in $\rdom$ is either $x$-, $y$-, or $xy$-monotone.

Here we partition the domain into eight disjoint regions using eight $xy$-monotone paths
as follows. We define the eight $xy$-monotone paths from $S$ in a way slightly different
to the one in Section~\ref{sec:eightpath}.
Consider the horizontal ray emanating from $s=p_1$ going rightwards,
and let $P\in\rectset_P$ be the polygon such that $\bbox(P)$ is the first rectangle
hit by the ray among the rectangles, at point $b$ on its left side.
If the upper endpoint $q$ of the leftmost vertical side of $P$
lies above $b$, we set $p'_1$ to $b$ and continue with the vertical ray from $p'_1$ to $q$, and
continue along the boundary chain of $P$ from $q$ to the left endpoint 
$p_2$ of the topmost side of $P$ in clockwise order. 
Otherwise, the horizontal ray continues going rightwards until it hits $P$ at a point 
$b'$. Then we set $p'_1$ to $b'$ and
continue along the boundary chain of $P$ from $p'_1$ to the left endpoint $p_2$ of the topmost 
side of $P$ in clockwise order.
We repeat this process by taking the horizontal ray from
$p_2$ going rightwards. Then we obtain an $xy$-monotone path
$\piru(p)=(p=p_1,p_1',p_2,p_2',\ldots)$, by following the boundary chain of $P$ from $p'_i$ to $p_{i+1}$
in clockwise order.
Thus, $\piru(p)$ is an $xy$-monotone path from $p$
that alternates going horizontally rightwards and going vertically upwards.
We define eight $xy$-monotone paths $\pi_\alpha(p)$ as in Section~\ref{sec:eightpath}.
Using these eight $xy$-monotone paths,
we construct at most eight disjoint regions.

Using those regions, 
we compute a minimum-link shortest path from $S$ to the portion of $T$
contained in each region.
Let $T'$ be the portion of $T$ contained in $\xydom$.
The closest pair $(s,t)$ of $S$ and $T'$ consists of their endpoints.
We compute $\xydom(s,t)$ using the method in Section~\ref{sec:xymonotone}.
Observe that every shortest path from $S$ to $T'$ is contained in $\xydom(s,t)$.
With $O(n)$ baselines defined by the sides of $\bbox(P)$ and the boundary segments of $P$ 
incident to $\xydom$ for all $P\in\rectset_P$,
we can show that there is a minimum-link shortest path from $S$ to $T'$
which is aligned to the baselines using an argument similar to the proof of 
Lemma~\ref{lem:base-align}.
Hence, we can compute a minimum-link shortest path from $S$ to $T'$ 
in the same time and space as in Lemma~\ref{lem:xymonotone}.
Similarly, we can compute a minimum-link shortest path from $S$ to $T'$ 
for the portions $T'$ of $T$ contained in other regions.
When $T'$ is contained in $\xdom$,
a minimum-link shortest path may have some winders, each of which 
contains the topmost or the bottommost side of $P$ for a rectilinear polygon $P\in\rectset_P$.
This can be shown by an argument similar to the proof of Lemma~\ref{lem:winder}.
Thus we can compute $d(S,T')$ using the same plane sweep algorithm on $\xdom$,
and find the dividers
which are midpoints of the topmost or the bottommost side of $P$
as we do in Section~\ref{sec:xmonotone} in the same time and space stated in Lemma~\ref{lem:xmonotone}.

\begin{lemma}\label{lem:outside_pocket}
For two axis-aligned line segments $S$ and $T$ in $\rdom$
such that both $S$ and $T$ are disjoint from $B(P)$ for all $P\in\rectset_P$,
we can compute a minimum-link shortest path from $S$ to $T$ in $\rdom$
in $O(n \log n)$ time using $O(n)$ space.
\end{lemma}

\subsection{\texorpdfstring{$S$}{S} or \texorpdfstring{$T$}{T} intersecting bounding boxes}
Each horizontal or vertical line segment contained in $\rdom$ intersects at most two bounding 
boxes of polygons in $\rectset_P$, and thus it can be partitioned into at most three pieces, 
one disjoint from the rectangles of $\rectset$, and the other two, each contained in 
the bounding box of a polygon in $\rectset_P$.
This applies to $S$ and $T$. 
Thus, in order to find a minimum-link shortest path from $S$ to $T$,
we need to consider at most 9 pairs, each consisting of one piece of $S$ and one piece of $T$,
and find a minimum-link shortest path for each pair.
We can handle the pair consisting of the pieces of $S$ and $T$ disjoint from the rectangles of 
$\rectset$ using the method in Section~\ref{sec:bbox-disjoint}.
In this section we show how to handle the remaining 8 pairs.
Each such pair has at least one piece of $S$ or $T$ 
that is contained in the bounding box of a polygon in $\rectset_P$.

Without loss of generality, we assume that $S$ is contained in $\bbox(P)$ of $P\in\rectset_P$
in the following. Let $C_S$ be the the component among the connected components of 
$\bbox(P)\setminus \text{cl}(P)$ that contains $S$, where $\text{cl}(P)$
is the closure of $P$.

\subsubsection{\texorpdfstring{$T$}{T} intersecting \texorpdfstring{$C_S$}{CS}}
We first consider the case that $T\cap C_S\neq\emptyset$.
We assume that $T$ is a vertical line segment. The case that $T$ is a horizontal line segment
can be handled analogously.
Observe that any closest point of $T$ from $S$ lies on $T'=T\cap C_S$,
that is, the problem reduces to computing
a minimum-link shortest path from $S$ to $T'$ in the rectilinear polygon $C_S$.
To ease the description, we simply assume that $T$ is contained in $C_S$.
There exists a shortest path from $S$ to $T$ which is not $x$-, $y$-, or $xy$-monotone.
However, $C_S$ is the rectilinear polygon without holes,
so we can use the algorithm of Schuierer~\cite{schuierer1996},
which computes a minimum-link shortest path 
between two points in a rectilinear polygon.

\begin{lemma}\label{lem:closest}
If no axis-aligned line segment contained in $C_S$ connects $S$ and $T$,
the closest pair of $S$ and $T$ is unique.
\end{lemma}
\begin{proof}
We show that the closest points of $S$ from any points of $T$ are the same.
Then by symmetry, the closest points of $T$ from any points of $S$ are the same,
and thus the lemma holds.
Assume to the contrary that there are two distinct closest points $s_1$ and $s_2$ 
in $S$ from two closest points $t_1$ and $t_2$ in $T$, possibly $t_1=t_2$, respectively.
Let $\pi_1$ and $\pi_2$ be the shortest paths such that $\pi_1$ connects $s_1$ and $t_1$, and 
$\pi_2$ connects $s_2$ and $t_2$. Clearly, both $\pi_1$ and $\pi_2$ are contained in $C_S$.

Since $S$ is vertical, the segments of $\pi_1$ and $\pi_2$ incident to $s_1$ and $s_2$ are 
horizontal, respectively.
Let $s'$ be the point in $S$ such that $y(s') = (y(s_1) + y(s_2)) / 2$,
and $H$ be the maximal horizontal segment contained in $C_S$ that contains $s'$.
Since no axis-aligned line segment contained in $C_S$ connects $S$ and $T$,
$H$ does not intersect $T$ but it intersects $\pi_1$ or $\pi_2$ at a point $p$.
Then we can get a shorter path from $S$ to $T$
by replacing the subpath from $s_1$ to $p$ of $\pi_1$ (or from $s_2$ to $p$ of $\pi_2$) 
with the segment $s'p$, a contradiction.
\end{proof}

If there is an axis-aligned line segment in $C_S$ connecting $S$ and $T$,
the line segment is a minimum-link shortest path.
Otherwise, for a point $t\in T$, we find $d(s,t)$ for every intersection point $s$ 
of $S$ and the horizontal baselines of $\rdom$.
Lemma~\ref{lem:base-align} also holds in $\rdom$,
so one of those intersection points is 
the closest point of $S$ from $T$.
Let $s^*\in S$ be the point achieving $d(s^*,t) = \min_{s} d(s,t)$.
Then $s^*$ is the closest point of $S$ from $T$ by Lemma~\ref{lem:closest}.
From $s^*$, we find the point $t^*\in T$ achieving $d(s^*,t^*) = \min_{t} d(s^*,t)$ 
among all intersection points $t$ of $T$ and the horizontal baselines of $\rdom$.
Finally we find two points $s^*$ and $t^*$,
so we can compute $\ml(s^*,t^*)$ using the data structure of Schuierer~\cite{schuierer1996} directly.

We compute the bounding boxes of the polygons in $\rectset_P$ and $C_S$ in $O(n)$ time.
We construct the data structure of Schuierer~\cite{schuierer1996} with 
$O(n)$ time and space for a rectilinear polygon $Q$ with $n$ edges that 
given two points $p$ and $q$ in $Q$, reports $d(p,q)$ and $\ml(p,q)$ in $O(\log n)$ query time,
and a minimum-link shortest path from $p$ and $q$ in $O(\log n + K)$ time,
where $K$ is the number of links of the path.
Since there are $O(n)$ baselines in $\rdom$, we can find $s^*$ and $t^*$ in $O(n \log n)$ time
using the data structure, and a minimum-link shortest path from $s^*$ to $t^*$ in $O(n)$ time since $K=O(n)$.

\subsubsection{\texorpdfstring{$T$}{T} disjoint from \texorpdfstring{$C_S$}{CS}}
Consider the case that $T$ is disjoint from $C_S$. 
The portion of the boundary of $C_S$ which is not incident to $P$ 
consists of a horizontal segment $H_S$ and a vertical segment $V_S$. 
We assume that $T$ is also contained in a connected component $C_T$ of 
$\bbox(P')\setminus \text{cl}(P')$ for a polygon $P'\in\rectset_P$.
Let  $H_T$ and $V_T$ for $T$ be the horizontal segment and a vertical segment
of the portion of the boundary of $C_T$ which is not incident to $P'$.

We compute minimum-link shortest paths from $S$ to $T$ passing through 
$H_S\cup V_S$ and $H_T\cup V_T$, and then we choose the optimal path 
among them.
In the following, we show how to compute a minimum-link shortest path from $S$ to $T$ passing through 
$V_S$ and $V_T$. The other cases can be handled analogously.
If no axis-aligned line segment contained in $C_S$ connects $S$ and $V_S$,
the closest pair $(s^*,v^*)$ of $S$ and $V_S$ is unique by Lemma~\ref{lem:closest}.
Thus, $d(S,v) = d(s^*,v) = d(s^*,v^*)+d(v^*,v)$ for any point $v\in V_S$.
Similarly, the closest pair $(t^*,u^*)$ of $T$ and $V_T$ is also unique if
no axis-aligned line segment contained in $C_T$ connects $T$ and $V_T$.
In this case we have $d(T,u) = d(t^*,u) = d(t^*,u^*)+d(u^*,u)$ for any point $u\in V_T$.


\begin{lemma}\label{lem:vconnect}
If the closest pair $(s^*,v^*)$ of $S$ and $V_S$ is unique, and
there is a shortest path from $S$ to $T$ passing through $V_S$,
there is a shortest path from $S$ to $T$ passing through $v^*$.
\end{lemma}
\begin{proof}
Let $\pi$ be a shortest path from $S$ to $T$ that
passes through a point $v\in V_S\setminus\{v^*\}$.
Since $d(s^*,v) = d(s^*,v^*)+d(v^*,v)$, we have
$|\pi|=d(S,T)=d(s^*,v)+d(v,T) = d(s^*,v^*)+d(v^*,v)+d(v,T)$.
Let $\pi^*$ be a path from $S$ to $T$ consisting of a shortest path from $s^*$ to $v^*$ and 
a shortest path from $v^*$ to $T$.
Since $d(v^*,T)\leq d(v^*,v)+d(v,T)$, 
we have $|\pi^*|=d(s^*,v^*)+d(v^*,T)\leq |\pi|=d(S,T)$.
Thus, $\pi^*$ is also a shortest path from $S$ to $T$.
\end{proof}

We compute a minimum-link shortest path from $S$ to $T$ as follows.
Let $Q_S$ be the set of intersection points of $V_S$ with the horizontal baselines
in $\rdom$, and $Q_T$ be the set of intersection points of $V_T$ with the horizontal baselines
in $\rdom$.
Let $\ml_H(X,Y)$ is the minimum number of links of all shortest paths connecting 
two sets $X$ and $Y$ whose segments incident to $Y$ are horizontal.
We first compute $d(S,v)$ and $\ml_H(S,v)$ for every point $v\in Q_S$.
We also compute $d(T,u)$ and $\ml_H(T,u)$ for every point $u\in Q_T$.
By Lemma~\ref{lem:vconnect},
once we have the unique closest pairs $(s^*,v^*)$ and $(u^*,t^*)$,
their distances $d(s^*,v^*), d(v^*,u^*), d(u^*,t^*)$,
and their minimum numbers of links $\ml_H(s^*,v^*), \ml_H(v^*,t^*)$,
we can compute a minimum-link shortest path $\pi$ from $S$ to $T$ 
passing through $s^*,v^*,u^*$ and $t^*$ in order.
Note that we do not guarantee that $\pi$ is a minimum-link shortest path from $S$ to $T$.
However, we can compute a minimum-link shortest path
while we compute $\ml_H(v^*,t^*)$ as follows.

Once we have $v^*$, $u^*$, and $d(v^*,u^*)$, we apply the algorithm in Section~\ref{sec:bbox-disjoint}.
In the algorithm, we construct $xy$-monotone subregions.
Let $\dom(v^*)$ and $\dom(u^*)$ be the
$xy$-monotone subregions incident to $v^*$ and $u^*$, respectively.
We may have $\dom(v^*)=\dom(u^*)=\xydom(v^*,u^*)$ 
if a shortest path from $v^*$ to $u^*$ is $xy$-monotone.

Consider a point $v\in Q_S$ that is incident to $\dom(v^*)$.
Then $d(s^*,v) = d(s^*,v^*)+d(v^*,v)$, $v$ lies on the outer boundary 
of $\dom(v^*)$, and $d(v^*,v)+d(v,t^*) = d(v^*,t^*)$.
Thus, we have $d(s^*,v)+d(v,t^*) = d(s^*,v^*)+d(v^*,t^*)$.
Once $\ml_H(v,t^*)$ is computed for every point $v$ of $Q_S$ 
that is incident to $\dom(v^*)$,
we can find a minimum-link shortest path from $S$ to $T$.
If $v\in Q_S$ is not incident to $\dom(v^*)$,
we have $d(v^*,v)+d(v,t^*) > d(v^*,t^*)$, and thus no shortest path from $S$ to $T$ passes through $v$.
This observation can also be applied for points in $Q_T$ that are incident to $\dom(u^*)$.
The plane sweep algorithm starts with updating $M(i)$'s for the horizontal baselines $H_i$ 
intersecting the vertical line segment $V$ of the outer boundary of $\dom(u^*)$ corresponding 
to the \origi event. At the \origi event, those $M(i)$'s are initialized to $\ml_H(t^*,V\cap H_i)$.
Observe that every intersection point $V\cap H_i$ is in $Q_T$.
It also computes $M(i)$'s for the horizontal baselines $H_i$ intersecting the vertical line segment
of the outer boundary of $\dom(v^*)$ corresponding to the \termi event.
Hence one of $M(i)$'s corresponds to $\ml_H(t^*,v^*)$ at the \termi event.
By choosing the minimum of $\ml_H(s^*,v) + \ml_H(v,t^*)-1$
for all $v\in Q_S$ incident to $\dom(v^*)$,
we finally obtain $\ml(S,T)$, 
and compute a minimum-link shortest path from $S$ to $T$.
Recall that to reduce the time complexity to $O(n \log n)$,
we do not maintain $M(i)$'s explicitly,
but focus on the minimum of $M(i)$'s using $O(\log n)$ nodes of $\segtree$ as we do in Section~\ref{sec:reduce-time}.
However, we observe that $\ml_H(s^*,v^*) + c = \ml_H(s^*,v)$ for every point $v$ of $Q_S$, where $c\in\{0,1,2\}$.
Therefore, by storing for each node $w$ of $\segtree$,  $\ml(w), U(w)$, and the second minimum and the third 
minimum of $M(i)$'s for $i\in[\alpha(w),\beta(w)]$,
we can compute a minimum-link shortest path from $S$ to $T$ without increasing time and space complexities.

If the closest pair of $S$ and $V_S$ is not unique,
there is a maximal line segment $V'\subseteq V_S$
such that for every point $p\in V'$, the shortest path from $S$ to $p$ is 
a horizontal line segment in $C_S$.
Recall that our algorithm uses the point $v^*$ in $V_S$ if the closest pair $(s^*,v^*)$ of $S$ and $V_S$ is unique.
Hence, instead of using $v^*$, we apply the algorithm using $V'$
and then we can compute a minimum-link shortest path.

Again using the data structure of Schuierer~\cite{schuierer1996},
we can compute $d(S,v)$ and $\ml_H(S,v)$ for all $v\in Q_S$ (and $d(T,u)$ and $\ml_H(T,u)$ for all $u\in Q_T$) in $O(n\log n)$ time.
Then we use the algorithms in Section~\ref{sec:bbox-disjoint} based on the methods in Sections~\ref{sec:xymonotone} and~\ref{sec:xmonotone} to compute $d(v^*,u^*)$.
Observe that the time and space complexities remain the same as stated in Lemma~\ref{lem:outside_pocket}. 
The initialization of $M(i)$'s at the \origi event of $\dom(u^*)$,
and the computation of $\ml_H(T,u^*)$ using $M(i)$'s at the \termi event of $\dom(v^*)$ do not affect the time
and space complexities asymptotically.
Therefore, we have the following theorem.

\begin{theorem}\label{thm:seg-result}
Given two axis-aligned line segments $S$ and $T$ in 
a box-disjoint rectilinear domain with $n$ vertices
in the plane,
we can compute the minimum-link shortest path from $S$ to $T$ in $O(n\log n)$ time
using $O(n)$ space.
\end{theorem}

\section{Extending to two polygons \texorpdfstring{$\textsf{S}$}{S} and \texorpdfstring{$\textsf{T}$}{T}}
Now we consider two rectilinear polygons $\textsf{S}$ and $\textsf{T}$ with $N$ vertices in $\rdom$.
We can compute a minimum-link shortest path from $\textsf{S}$ to $\textsf{T}$
using our algorithms in previous sections.
Since $\textsf{S}$, $\textsf{T}$, and obstacles are pairwise box-disjoint, the distance $d(\textsf{S},\textsf{T})$ between $\textsf{S}$ and $\textsf{T}$ can be represented as
$d(\textsf{S},\textsf{T}) = \min_{s\in B(\textsf{S}), t\in B(\textsf{T})}\{d(s,t)+\min_{s'\in \textsf{S}}d(s,s')+\min_{t'\in \textsf{T}}d(t,t')\}$.
If we construct the $L_1$ Voronoi diagram of $N$ boundary segments of $\textsf{S}$ (or $\textsf{T}$)~\cite{papadopoulou2001} in $O(N \log N)$ time using $O(N)$ space,
we can maintain and report $\min_{s'\in \textsf{S}}d(s,s')$ and $\min_{t'\in \textsf{T}}d(t,t')$
for any $s\in B(\textsf{S})$ and $t\in B(\textsf{T})$
in $O(\log N)$ query time.
From this observation, together with Lemma~\ref{lem:monotone},
we have the following lemma.

\begin{lemma}\label{ob:polygon_same_monotone}
If there is an $x$-monotone shortest path from $\textsf{S}$ to $\textsf{T}$,
then every shortest path from $\textsf{S}$ to $\textsf{T}$ is $x$- or $xy$-monotone.
If there is a $y$-monotone shortest path from $\textsf{S}$ to $\textsf{T}$,
then every shortest path from $\textsf{S}$ to $\textsf{T}$ is $y$- or $xy$-monotone.
\end{lemma}

From Lemma~\ref{ob:polygon_same_monotone}, we can partition the box-disjoint rectilinear domain into
eight disjoint regions using eight $xy$-monotone paths from $B(\textsf{S})$
as done in Section~\ref{sec:bbox-disjoint}. See Figure~\ref{fig:boxdisjointpolygon}(a).
There are $O(N)$ vertical and horizontal baselines defined by the boundary segments of 
$\textsf{S}$ and $\textsf{T}$, and Lemma~\ref{lem:base-align} also holds.
Thus, we compute a minimum-link shortest path aligned to the baselines
of each region in which the portion of $\textsf{T}$ is contained.

\begin{figure}[t]
  \begin{center}
    \includegraphics[width=0.9\textwidth]{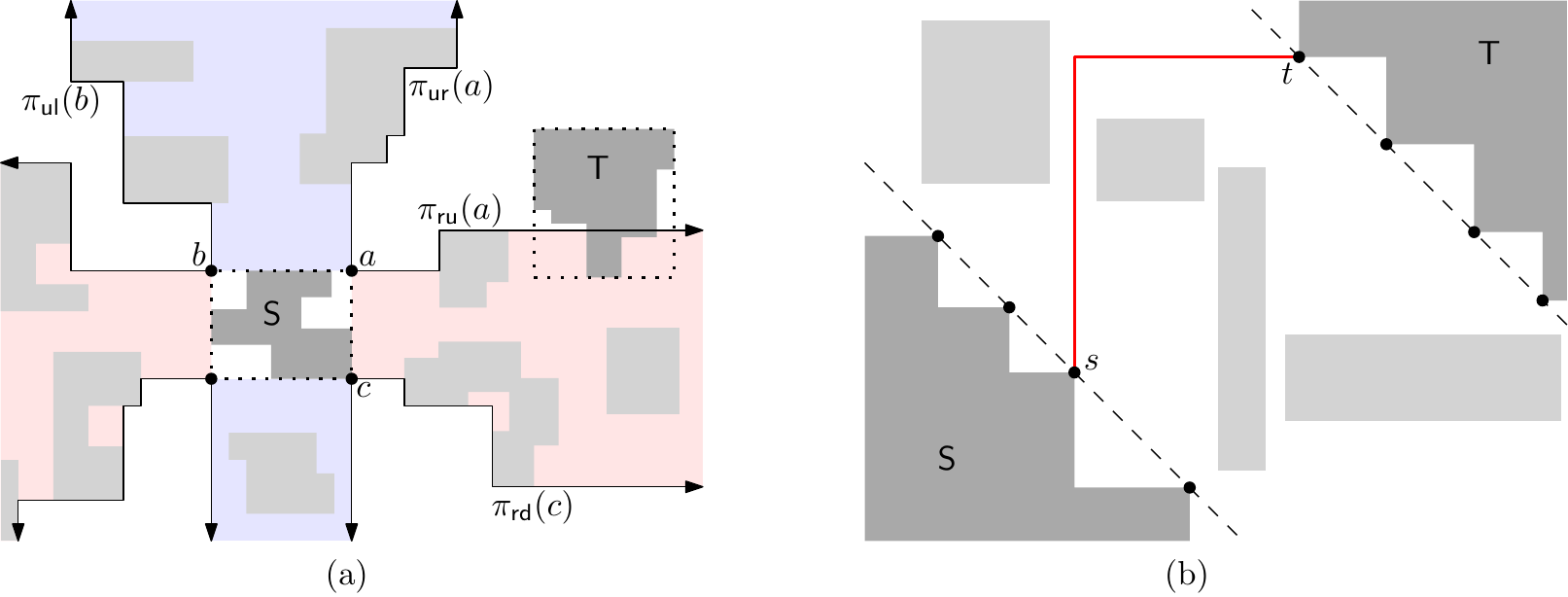}
    \caption{\small
    (a) Eight regions of $\rdom$ by eight $xy$-monotone paths from four corners of $B(\textsf{S})$ with box-disjoint obstacles.
$\textsf{T}$ intersects at most five regions.
(b) There are nine closest pairs of $\textsf{S}$ and $\textsf{T}$.
Among all paths connecting closest pairs, the minimum-link shortest path from $s$ to $t$ is the optimal.
    }
    \label{fig:boxdisjointpolygon}
  \end{center}
\end{figure}

Let $\textsf{T}'$ be the portion of $\textsf{T}$ contained in $\xydom$.
Since $\textsf{S}$ and $\textsf{T}'$ are rectilinear polygons,
there can be more than one closest pair of points for $\textsf{S}$ and $\textsf{T}'$.
Moreover, the points appearing in the closest pairs are on line segments with slopes $\pm 1$. See Figure~\ref{fig:boxdisjointpolygon}(b).
If we compute $\xydom(s,t)$ for every closest pair $(s,t)$ of $\textsf{S}$ and $\textsf{T}'$,
the time and space complexities may increase.
Instead, 
we modify the plane sweep algorithm slightly from that in Section~\ref{sec:xymonotone}.
There can be more than one \origi and \termi events during the plane sweep
because there can be more than one closest pair of $\textsf{S}$ and $\textsf{T}'$.
Also, there are no \attac and \detac events since we do not compute $\xydom(s,t)$.
For each event $E_j$ in Section~\ref{sec:xymonotone},
however, we use $\alpha(j)$, $\beta(j)$ and $\rangebbst$ to maintain active baselines.
Since we have all points in closest pairs of $\textsf{S}$ and $\textsf{T}$,
we set two horizontal baselines with the smallest and largest $y$-coordinate respectively from those points.
The horizontal baselines between the two baselines are used for inserting the range, which represents active baselines,
into $\rangebbst$ of each event.
Since there are $O(N+n)$ baselines between the two baselines,
the time to handle an event takes $O(\log (N+n))$ time.
Also, there are additional $O(N)$ \origi and \termi events with $O(n)$ the other events,
so we can compute a minimum-link shortest path from $\textsf{S}$ to $\textsf{T}'$
in $O((N+n)\log (N+n))$ time using $O(N+n)$ space.

Let $\textsf{T}'$ be the portion of $\textsf{T}$ contained in $\xdom$.
Every shortest path from $\textsf{S}$ to $\textsf{T}'$ is $x$-monotone,
so we can compute the closest pairs of $\textsf{S}$ and $\textsf{T}'$
as the sweep line encounters each vertical line segments of $\textsf{T}'$
using the plane sweep algorithm in Section~\ref{sec:xmonotone}.
Then we can compute dividers in the same way without modifying the algorithm in Section~\ref{sec:xmonotone}.
Lemmas related to dividers in Section~\ref{sec:xmonotone} still hold,
so we can compute a minimum-link shortest path connecting dividers similarly.
As above, we can compute a minimum-link shortest path from $\textsf{S}$ to a divider (or from a divider to $\textsf{T}'$).
This implies we obtain a minimum-link shortest path from $\textsf{S}$ to $\textsf{T}'$.
We omit the details.

Therefore, we have Theorem~\ref{thm:result}.

\section{Concluding Remarks}
We propose the algorithm to compute a minimum-link shortest path connecting two rectilinear polygons
in the box-disjoint rectilinear domain efficiently.
Our algorithm computes a minimum-link shortest path from a point to the line segment using plane sweep paradigm,
based on the monotonicity of the optimal path.
Then we can extend objects to rectilinear polygons and apply the slightly modified algorithm.

Still there are many problems to be considered.
One typical problem is to compute a minimum-link shortest path connecting two objects
in a general rectilinear domai such that the obstacles in the domain are not necessarily box-disjoint.
There is a previous work in a general rectilinear domain, but the result does not seem to have the optimal time and space complexities.

\bibliographystyle{plain}
\bibliography{paper}

\end{document}